\ifdefined\isarxiv
\documentclass[10pt]{article}
\usepackage[margin=1.5in]{geometry}
\usepackage{times}
\usepackage{hyperref}
\else
\documentclass{article}
\usepackage[nonatbib,preprint]{neurips_2023}
\fi

\usepackage{multicol,enumitem}
\usepackage{hyperref, amsmath}
\DeclareMathOperator*{\argmax}{arg\,max}

\usepackage{hyperref}       % hyperlinks
\usepackage{url}   

\usepackage{amsmath}
\usepackage{amssymb}
\usepackage{amsthm}

\usepackage{booktabs}       % professional-quality tables
\usepackage{amsfonts}       % blackboard math symbols
\usepackage{nicefrac}       % compact symbols for 1/2, etc.
\usepackage{microtype}      % microtypography
\usepackage{xcolor}         % colors
\usepackage{amsmath,amsthm,amssymb}
\usepackage{graphicx}
\usepackage{cleveref}
\usepackage{comment}
\usepackage{bbm}
\usepackage{textcomp}
\usepackage{wrapfig}
\usepackage{float}

\usepackage{url}            % simple URL typesetting
\usepackage{booktabs}       % professional-quality tables
\usepackage{amsfonts}       % blackboard math symbols
\usepackage{nicefrac}       % compact symbols for 1/2, etc.
\usepackage{microtype}      % microtypography
\usepackage{multirow}
\usepackage{algorithm}
\usepackage{algorithmic}
\usepackage{amsmath}
\usepackage{amsthm}
\usepackage{amssymb}
\usepackage{color}
\usepackage[english]{babel}
\usepackage{graphicx}
\usepackage{grffile}
\usepackage{wrapfig,epsfig}
\usepackage{url}
\usepackage{color}
\usepackage[T1]{fontenc}
\usepackage{bbm}
\usepackage{comment}
\usepackage{tikz}
\usepackage{pgfplots}
\pgfplotsset{compat=1.8}
\tikzset{elegant/.style={smooth,thick,samples=500,magenta}}

\theoremstyle{plain}
\newtheorem{theorem}{Theorem}[section]
\newtheorem{lemma}[theorem]{Lemma}
\newtheorem{remark}[theorem]{Remark}
\newtheorem{corollary}[theorem]{Corollary}
\newtheorem{proposition}[theorem]{Proposition}
\theoremstyle{definition}
\newtheorem{definition}[theorem]{Definition}

\newtheorem{assumption}[theorem]{Assumption}

\newtheorem{example}[theorem]{Example}

\newtheorem{claim}[theorem]{Claim}
\crefname{assumption}{Assumption}{Assumptions}

\theoremstyle{plain}
\newtheorem*{thm*}{Theorem}

\theoremstyle{plain}

\newcommand{\fed}{\mathrm{fed}}

\newcommand{\tr}{\mathrm{tr}}

\newcommand{\eff}{\mathrm{eff}}
\newcommand{\peff}{\mathrm{pure}}
\newcommand{\supp}{\mathrm{supp}}
\newcommand{\sg}{\mathrm{SG}}
\newcommand{\poa}{\mathrm{POA}}

\newcommand{\R}{\mathbb{R}}
\newcommand{\E}{\mathbb{E}}

\newcommand{\M}{\mathbf{M}}
\newcommand{\0}{\mathbf{0}}

\newcommand{\X}{\mathcal{X}}
\newcommand{\mec}{\mathcal{M}}

\newcommand{\wt}{\widetilde}

\definecolor{b2}{RGB}{51,153,255}
\definecolor{mygreen}{RGB}{80,180,0}

\title{Evaluating and Incentivizing Diverse Data Contributions in Collaborative Learning}

\ifdefined\isarxiv
\author{
}
\date{}

\else

\author{%
  Baihe Huang\\
  University of California, Berkeley\\
  \texttt{baihe\_huang@berkeley.edu} \\
  \And
  Sai Praneeth Karimireddy\\
  University of California, Berkeley\\
  \texttt{sp.karimireddy@berkeley.edu} \\
  \AND
  Michael I. Jordan\\
  University of California, Berkeley\\
  \texttt{jordan@cs.berkeley.edu} \\
}

\fi
\begin{document}

\ifdefined\isarxiv
%\begin{titlepage}
  \maketitle
  \begin{abstract}
  \end{abstract}
 % \thispagestyle{empty}
%\end{titlepage}

%\pagenumbering{roman}
%{\small
%\hypersetup{linkcolor=black}
%\tableofcontents
%}
%\newpage
\else
\maketitle

\begin{abstract}
For a federated learning model to perform well, it is crucial to have a diverse and representative dataset. However, the data contributors may only be concerned with the performance on a specific subset of the population, which may not reflect the diversity of the wider population. This creates a tension between the principal (the FL platform designer) who cares about global performance and the agents (the data collectors) who care about local performance. 
In this work, we formulate this tension as a game between the principal and multiple agents, and focus on the linear experiment design problem to formally study their interaction. We show that the statistical criterion used to quantify the diversity of the data, as well as the choice of the federated learning algorithm used, has a significant effect on the resulting equilibrium. We leverage this to design simple optimal federated learning mechanisms that encourage data collectors to contribute data representative of the global population, thereby maximizing global performance.
\end{abstract}

\fi

\section{Introduction}

Collaborative learning can be viewed as a transactional process where participants collectively receive a reduction in uncertainty in return for sharing their data~\cite{karimireddy2022mechanisms}. However, each participant may be primarily interested in inference for a different sub-population of the global population. Thus a reduction in uncertainty on the latter may not necessarily translate to an improvement for every participant.

Consider in particular a collaborative learning project between multiple countries to study rare cancers~\cite{moncada2020vantage6,geleijnse2020prognostic}. Different countries operate cancer registries with the goal of collecting comprehensive data on rare cancer cases within their jurisdictions. These registries collaborate to pool their data and resources. However, each registry has the responsibility to prioritize the benefit for their own population while minimizing the risks associated with data collection and sharing. Thus, the global performance needs to be balanced with the specific needs and goals of each registry.

The need to balance local and global interests becomes even more critical when collecting data from marginalized communities. Issues of equity and autonomy underpin indigenous critiques of genetic research and the sharing of genomic data~\cite{hudson2020rights,chediak2020sharing}. Such communities have historically faced exploitation and mis/under-representation in research studies~\cite{graham2015disparities,albain2009racial,nana2021health}. Therefore, it is essential to  consider carefully the costs incurred by and benefits provided to them individually.

\begin{figure}[h]
    \centering
    \includegraphics[width=\linewidth]{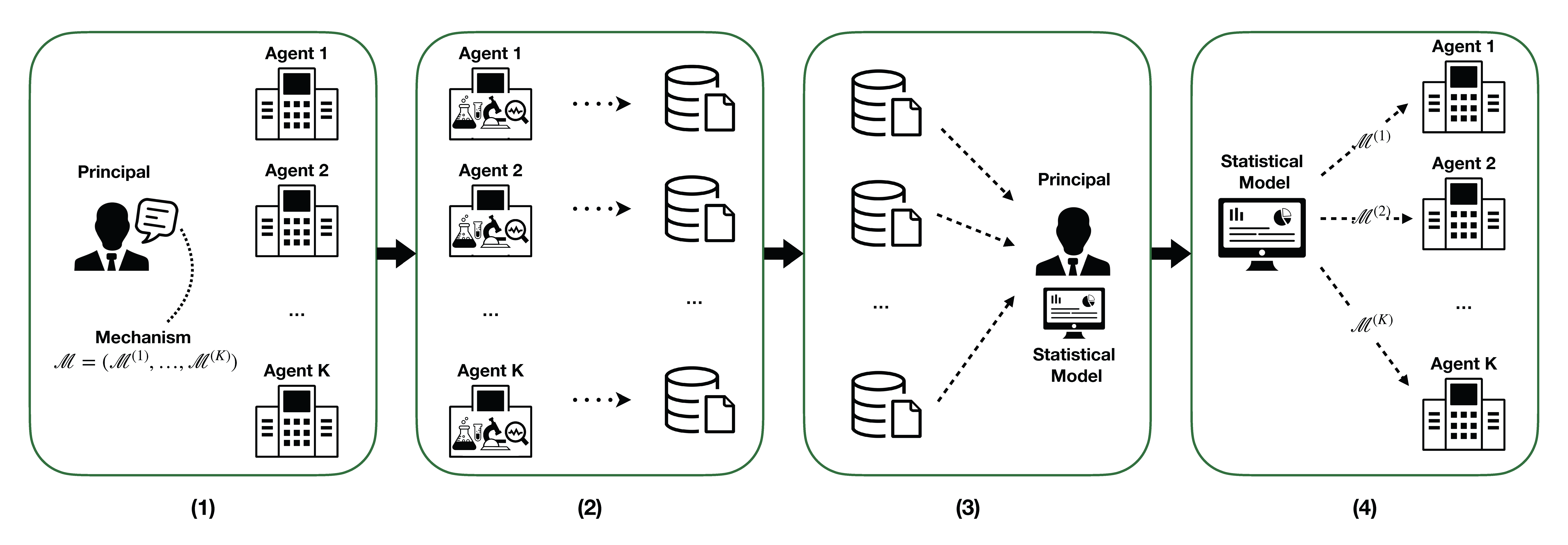}
    \vspace{-2.5em}
    \caption{Principal-agent experiment design. 
    (1) The principal publishes a mechanism incentivizing participation. (2) Each agent strategically selects how many and which of their available experiment conditions to collect and share. (3) The agents engage in collaborative learning with the principal, who utilizes the collected data points to train a statistical model. (4) The principal applies the mechanism to the trained model for each agent and subsequently distributes the models to them.}
    \label{fig:flowchart}
\end{figure}

We formalize this problem as a game between a principal (the platform designer), and multiple agents (participants) whose needs and agency needs to be respected---see Fig.~\ref{fig:flowchart}. Together, they wish to determine a statistical model between responses and variables. Each agent has access to a set of experiment conditions relevant to specific demographic groups within their population. They autonomously decide how many (as well as which) samples to collect and share. The platform then employs federated learning to train a model on the collective data, which is then shared back to the agents. Notably, each agent wishes to minimize the data costs incurred while maximizing uncertainity reduction.

This can be seen as a "multi-agent" version of the classic \emph{optimal experiment design} problem~\cite{wald1943efficient,kiefer1958nonrandomized,kiefer1960equivalence,karlin1966optimal,atwood1969optimal,fedorov1972optimal}, where the final allocation of samples among experiment conditions results from decisions made by multiple agents. 
Unlike classical theories, we introduce game-theoretic subtleties since each agent is primarily concerned with the validity of the model for its specific demographic group. 
Hence, we must account for the strategic behavior that emerges due to both data diversity and cost heterogeneity. 
In this context, two fundamental questions arise. 

The first question is efficiency---it is crucial to allocate resources such as time, money, and materials in the most efficient and informative manner. Such global efficiency was the main concern of classic theory of optimal experiment design which proposed different optimality (efficiency) criterion. In our version, we ask:

\begin{center}
\textit{When is it in the agents' best interest to follow the globally efficient optimal experiment design?}
\end{center}

The second issue is maximizing the amount of information collected. Conventional data-sharing mechanisms such as federated learning face the critical issue of free-riding \cite{baumol2004welfare,choi2019contributors,sarikaya2019motivating,lin2019free,ding2020incentive,sim2020collaborative,xu2021gradient}, where some strategic agents may contribute minimal or no data but still benefit from an improved model. The principal may instead want to maximize the information generated by data contributions from all agents (without regards to efficiency). This raises the following challenge for the principal:

\begin{center}
\textit{Can the principal design mechanisms to incentivize strategic agents to contribute their fair share of data, thereby maximizing the information produced?}
\end{center}

The quality and diversity of data play a vital role in the context of heterogeneous experiment conditions. In this study, we address the aforementioned challenges by specifically focusing on linear experiment design \cite{pukelsheim2006optimal,silvey2013optimal}, where diversity is characterized by the Fisher information matrix and quality is assessed using optimality criteria. The main contributions of this paper are as follows:
\vspace{-0.5em}
\begin{itemize}
    \item We formalize the problem of \emph{principal-agent experiment design} to capture the strategic behavior of self-interest agents in a collaborative learning process where the collection of data is planned and optimized. 
    Notably, our model considers both the quantity and diversity of the data, extending the scope of previous work \cite{karimireddy2022mechanisms}.
    \item 
    Within this framework, we investigate the statistical efficiency of the multi-agent system in the federated learning mechanism. Surprisingly, we demonstrate that the D-criterion is the \emph{only} commonly used optimality criterion that supports efficiency in the multi-agent system. This finding provides strong support for the application of the D-criterion in collaborative experiment design problems.
    \item We show that strategic agents may display excessive free-riding in heterogeneous data settings. Motivated by this fact, we consider the Stackelberg game between the principal and multiple agents, and propose a mechanism that incentives the agents to generate fair amounts of data so that the total information is maximized.
\end{itemize}

\paragraph{Related work}
In recent years, Federated Learning (FL) \cite{konevcny2016federated,mcmahan2017communication, kim2019blockchained, kairouz2021advances, li2020review, mancini2021data} has emerged as an important machine learning paradigm that allows multiple distributed clients to train a central statistical model under the orchestration of a principal. The cooperative nature of the framework raises economical and ethical concerns such as free-riding and fairness \cite{baumol2004welfare,fraboni2021free,mohri2019agnostic,huang2020exploratory,shi2021survey}. 
To address with these concerns, several authors have investigated free-rider attacks and have developed methods for detection~\cite{richardson2020budget, sarikaya2019motivating,lin2019free,fraboni2021free,ding2020incentive, zhang2022enabling}. Another line of works designs metrics for quantifying the contribution of each agent~\cite{ghorbani2019data,jia2019towards,wang2020principled}.
More closely related to the current paper are recent works that apply the theory of contracts and incentives~\cite{smith2004contract, laffont2009theory, bolton2004contract} to FL. In particular, \cite{tian2021contract} propose a mechanism to achieve improved generalization accuracy by eliciting the private type of the agents, \cite{sim2020collaborative, xu2021gradient} propose mechanisms based on notions from the cooperative game theory literature to incentivize agents through the model quality, and \cite{karimireddy2022mechanisms} introduce mechanisms based on accuracy-shaping to maximize the number of data points generated by each agent. Recently, \cite{xu2023fair} propose a collaborative data generation method that satisfies the Individual Rationality and fairness.

Our work is different from these works in the following way: (i) we study an autonomous data-generation process in which each agent can strategically choose what experiment condition to collect data from, and (ii) we model the utility of self-interested agents and formulate the process of data generation as a game. 
Addressing these challenges requires a more thoroughgoing blend of economic and statistical foundations, and we pursue this blend within the statistical theory of experiment design \cite{pukelsheim2006optimal,silvey2013optimal}. 
In particular, we repose on a long line of works studying optimally in linear experiment design~\cite{smith1918standard,wald1943efficient,kiefer1958nonrandomized,eccleston1974theory,kiefer1975construction,cheng1978optimality}. A key concept in this literature is the notion of the D-criterion as a formalization of optimality~\cite{shah1989optimality}. Our work makes use of the D-criterion, in particular its instantiation in the work of \cite{kiefer1960equivalence,sibson1973da}. 
We bring this line of work on optimal experiment design into contact with the study of multi-agent learning systems.

\paragraph{Notations}
For any vector $w \in \R^n$ and index set $G \subset [n]$, let $w_{G} \in \R^{|G|}$ denote the vector formed by the coordinates of $w$ in the index set $G$ (preserving the order),
and let $G^c$ denote the complement of $G$. 
We use $d F(u,v)$ to denote the Gateaux derivative of $F$ at $u$ in the direction $v$. 
Define $(x)_+ := \max\{x,0\}$. 
Let $\R^d_+$ denote the nonnegative orthant in $\R^d$, i.e. $\R^d_+ = \{x\in \R^d: x_i \geq 0, \forall i \in [d]\}$. 
Denote $\langle A,B \rangle := \mathrm{tr}[A^\top B]$ for $A,B \in \R^{d \times d}$. 
Finally, $M^\dagger$ represents the Moore-Penrose inverse of matrix $M$.

\section{Single Agent Experiment Design}

\subsection{Quantifying and minimizing uncertainty}
\label{sec:experiment_uncertainty}

Many scientific problems involve determining the underlying parameter $\theta$ in relationships of type%\vspace{-1em}
\begin{align}\label{eq:def-problem}
    y = \theta^\top x + e\,,
\end{align}
where $x \in \R^d$ represents an experimental condition under which the data is collected, and $e$ has zero mean and unit variance. Given a set of observations, $(y_i,x_i)_{i  = 1,\dots,m}$, ordinary least squares (OLS) yields the estimator $\hat \theta = (X^\top X)^{\dagger}X^\top Y$ where $X= (x_1, \cdots, x_m)^\top \in \R^{m \times d}$ and $Y = (y_1,\dots,y_m)^\top \in \R^m$ denote the experimental conditions and the responses. The variability of $\hat \theta$ is governed by the Fisher information matrix $\mathcal{I}(X;\theta) = \sum_{i=1}^m x_ix_i^\top$. Specifically, the expected error for a prediction at the covariate $x$, conditioned on $X$, is
\begin{equation}\label{eq:uncertainity}
    \mathcal{E}(x;X) = \E \big(\hat\theta^\top x - \E(y | x)\big)^2 = x^\top \mathrm{Cov}(\hat \theta)x = x^\top \mathcal{I}^{\dagger} x\,.
\end{equation}    
The problem of optimal experiment design is to select the training data $X$ which ``maximize'' the information matrix $\mathcal{I}$, thereby minimizing uncertainty. More precisely, we are given a \textit{design space} $\X \subset\R^d$ containing {all possible data points} which could be collected, and we choose a sampling strategy $\pi$ (referred to as the \textit{design measure}), which is a measure over $\X$. For technical simplicity, we assume that the design space is finite with $\X = \{x_1,\dots,x_n\}$. We then define the \textit{information matrix} as a function of $\pi$:
\begin{align}\label{eq:def-information}
    \M(\pi) := \sum_{i=1}^n \pi_i x_ix_i^\top .
\end{align}
Then, similar to \eqref{eq:uncertainity}, if an experimenter draws $m$ samples $X_m$ following $\pi$, Slutsky's theorem implies the expected error satisfies the following asymptotic relationship
\begin{align}\label{eq:uncertainty_asymp}
    m \cdot \mathcal{E}(x;X_m) \overset{p}{\rightarrow} x^\top \M(\pi)^{-1} x\,.
\end{align}
Thus, $\M(\pi)^{-1}$ is a matrix representing uncertainty along different directions under a sampling strategy $\pi$. We reduce this to a single scalar using an \textit{optimality criterion} $f$\footnote{Without loss of generality, we assume $f(M^{-1}) = -\infty$ when $M$ is singular, and so restrict ourselves to nonsingular $\M(w)$.} which is a function from the set of symmetric matrices in $\R^{d \times d} \rightarrow \R$:
\begin{align}\label{eq:def-criterion}
    \max_{\pi} f(\M(\pi)^{-1}), &~ \text{s.t.} ~ \pi \in \Delta(\X)\,.
\end{align}
Some popular choices of optimality criterion are as follows:
\begin{multicols}{2}
\begin{itemize}[noitemsep, topsep=0pt, leftmargin=*]
\item E-criterion: $f_E(M^{-1}) = -\|M^{-1}\|_2$
\item A-criterion: $f_A(M^{-1}) = -\tr\left(M^{-1}\right)$
\item V-criterion: $f_V(M^{-1}) = -\frac{1}{|\X|}\sum_{x \in \X}\left(x^\top M^{-1} x\right)$.
\item D-criterion: $f_D(M^{-1}) = \log \det M$
\item G-criterion: $f_G(M^{-1}) = - \max_{x \in \X} \left(x^\top M^{-1} x\right)$
\end{itemize}
\end{multicols}

\subsection{Agent utility}\label{sec:agent_utility}

Each agent wants to minimize uncertainty (on their design space) while also minimizing the costs of data collection and sharing. Consider an agent $k$ with design space $\X_k \subseteq \X$. Here, $\X_k$ represents the sub-population that agent $k$ cares about and has access to. In general, this is different from the global $\X$. Further, define the index set $G_k \subseteq [n]$ such that $\X_k = \{x_i\}_{i \in G_k}$.

Agent $k$ may wish to collect multiple data points, and needs to decide the quantity of samples along with the sampling strategy. To model this, we define $w$ to be a general design measure over $\X$; that is, $w \in \R_+^{|\X|}$ and may not sum to one. Suppose that the cost of collecting a single data point is $c^{(k)}>0$. This can represent both the actual cost incurred in collecting and storing the data as well as potential privacy risks associated with storing and sharing it. Then, the total cost incurred by agent $k$ is $c^{(k)}\sum_{i \in G_k}w_i$.

We next need to quantify the information gained from a design measure $w$ for agent $k$. While Eq.~\eqref{eq:def-information} defines the information matrix for the global design space $\X$, agent $k$ only cares about $\X_k$, which is of rank $r_k$ (that may be less than $d$). We thus need to consider a \emph{local information matrix}~\cite{sibson1973da,silvey1978optimal} representing uncertainty along directions only in $\X_k$:
\begin{equation}\label{eq:def-local-information}
\M^{(k)}(w) := ({A^{(k)}}^\top\M(w)^{\dagger}A^{(k)})^{-1}\,.    
\end{equation}
Here, $A^{(k)} \in \R^{d \times r_k}$ such that $A^{(k)} {A^{(k)}}^\top$ is a projection matrix onto $\mathrm{span}(\X_k)$. To see the meaning of local information matrix, consider $X$ being a set of $\sum_{i=1}^m w_i$ samples from $\X$ where $x_i$ is collected $w_i$ times. Then for any $z \in \R^{r_k}$, the expected error for a prediction at the covariate $A^{(k)}z \in \R^d$ is given by
\begin{align*}
    \mathcal{E}(A^{(k)}z;X) = (A^{(k)}z)^\top \big(\textstyle\sum_{i=1}^n w_i x_ix_i^\top \big)^{\dagger} A^{(k)}z = z^\top \M^{(k)}(w)^{-1} z\,.
\end{align*}
Therefore comparing to Eq.~\eqref{eq:uncertainty_asymp}, $\M^{(k)}(w)$ quantifies the variability of prediction only in the subspace spanned by $\X_k$. 
The value of a design strategy $w$ for agent $k$ on $\X_k$ using criterion $f^{(k)}(\cdot)$ can then be written as
\(
    f^{(k)}\left({\M^{(k)}(w)}^{-1}\right).
\)
Here, the criterion $f^{(k)}$ depends on $k$ since it may implicitly depend on the design space $\X_k$. In particular, the G-criterion takes a maximum of prediction error over the design space $\X_k$, and the V-criterion takes an average. We refer to Appendix~\ref{sec:local_criterion} for the explicit expressions of such criteria.

Next, note that agent $k$ only has control over $w_i$ for $i \in G_k$; i.e., it can only decide the sampling strategy over $\X_k$. For convenience, define $w_{G_k} = (w_i)_{i \in G_k}$. Putting all of this together, we can define the goal of an agent to be to choose a design measure $w_{G_k}$ which maximizes its utility
\begin{equation}\label{eq:single_agent_utility}
    \max_{w_{G_k}}  \Big\{u^{(k)}(w) := f^{(k)}\left({\M^{(k)}(w)}^{-1}\right) - c^{(k)} \sum_{i\in G_k} w_i\Big\}.
\end{equation}
Finally, note that the cost incurred by agent $k$ is completely independent of $w_{{G}^c_k} := (w_j)_{j\notin G_k}$. However, the local information matrix $\M^{(k)}(w)$ depends on the whole $w$ and $\X$. In particular, if the complementary design space $x_j \in \X \setminus \X_k$ is similar to $\X_k$ and $w_j > 0$, then the local information matrix $\M^{(k)}(w)$ as well as $f^{(k)}\big({\M^{(k)}(w)}^{-1}\big)$ will be larger. Thus, $w_{{G}^c_k}$ represents free outside information given to agent $k$, and importantly it affects the optimal choice of $w_{G_k}$. 
The following result connects this formulation with the standard experiment design.
\begin{proposition}[Isolated agent follows optimal design]\label{cla:connection_to_optimal_design}
Suppose that $f$ satisfies the following: there exists $p \in \R$ and function $g: \R \to \R$ such that $f^{(k)}(\lambda M) = \lambda^p \cdot f^{(k)}(M) + g(\lambda)$ for all p.s.d. matrices $M$. Also, define the projected local design space $\tilde\X_k := \big\{{A^{(k)}}^\top x: x \in \X_k \big\}$. Then, when $w_{{G}^c_k} = 0$, the strategic response of agent $k$, i.e. design $w_{G_k}^*$ which maximizes $u^{(k)}$ in Eq.~\eqref{eq:single_agent_utility}, satisfies\vspace{-0.7em}
\[
w_{G_k}^* \propto \argmax_{\tilde\pi \in \Delta(\tilde\X_k)} f\big(\big(\textstyle\sum_{x_i \in \tilde\X_k}\tilde\pi_i x_i x_i^\top\big)^{-1}\big)\,.
\]
\end{proposition}
This result shows that in the absence of free outside information, a rational agent optimizing their utility Eq.~\eqref{eq:single_agent_utility} will follow the optimal design measure over the local design space. Thus, the single agent version of our problem recovers the classical experiment design problem. Note that all the common optimality criterion listed in Section~\ref{sec:experiment_uncertainty} satisfies the condition of this claim. 
When $w_{G_k^c}$ is non-zero, the optimal strategy might depend on the available outside information $\sum_{i \in G_k^c}w_ix_ix_i^\top$. 
In conclusion, the problem in Eq.~\eqref{eq:single_agent_utility} can be seen as a generalization of the problem in Eq.~\eqref{eq:def-criterion} where free outside information is given and $w_{G_k}$ need not to be a probability over $\X_k$.

\section{Principal-Agent Experiment Design}

In this section, we introduce our problem of principal-agent experiment design which brings together game-theoretic considerations and properties of optimal experiment designs in a principal-agent framework~\cite{laffont2009theory}. Our overall framework is summarized in Algorithm~\ref{alg:alg}.

\subsection{Collaborative experiment design mechanism}

We model the interaction between multiple agents and a coordinating principal. Consider $K$ self-interested agents. Each agent $k$ has a local design space $\X_k = \{x_i\}_{i \in G_k}$. The global design space is then $\X = \cup_{k \in [K]}\X_k = \{x_1, \dots, x_n\}$. For ease of presentation, assume that the indices are sorted such that $G_1, \dots, G_K$ form consecutive partitions of $[n]$. Then, $w = (w_{G_1}, \dots, w_{G_K})$ is the global design measure with each agent $k$ controlling $w_{G_k}$. 
The data contribution of agent $k$ can thus be summarized by $w_{G_k}$. 
\paragraph{Mechanism definition.} The principal is given access to the entire data contributions, i.e., the global design measure $w$. Then the principal sets up a \emph{mechanism} to assign a subset of this contribution to each agent $k$. More formally, we introduce the following definition. 
\begin{definition}[Contribution-assigning mechanism]
A mechanism $\mec$ is defined as:
\begin{equation}\label{eq:def-mechanism}
    \mec := \big(\mec^{(k)}: \R_+^{n} \rightarrow \R_+^{n}\big)_{k \in [K]}\ \text{ satisfying }\  \mec^{(k)}(w) \leq w\,.
\end{equation}
The inequality $\mec^{(k)}(w) \leq w$ applies element-wise for all $i \in [n]$.
\end{definition}
Thus, a mechanism represents a re-allocation of the design measure (and hence data) to each of the clients. The utility enjoyed by client $k$ under mechanism $\mec$ can be written as
\begin{equation}\label{eq:utility-under-mechanism}
 \big(u^{(k)} \circ \mec^{(k)}\big)\,(w) := f^{(k)}\left(\big({\M^{(k)}(\mec^{(k)}(w))}\big)^{-1}\right) - c^{(k)} \sum_{i\in G_k} w_i\,.   
\end{equation}
Here, $u^{(k)} \circ \mec^{(k)}$ represents the composition of the agent's utility function $u^{(k)}$, which depends solely on the agent's personal valuation, and the mechanism $\mec^{(k)}$ implemented by the principal.

\begin{remark}[Accuracy shaping]
    The mechanism Eq.~\ref{eq:def-mechanism} can be understood as shaping the accuracy \cite{karimireddy2022mechanisms} of the model that is sent to each agent. For example, the standard federated learning mechanism $\mec^{(k)}_{\fed}$ would distribute the global model to all the agents. This corresponds to setting $\mec^{(k)}_{\fed}(w) = w$ for all agents. To incentivize agents to contribute high-quality data, the mechanism may adjust the accuracy of the model depending on the data quality generated by each agent.
\end{remark}

\begin{algorithm}[!t]
\caption{Principal-Agent Collaborative Experiment Design}
\begin{algorithmic}[1]
\STATE The principal selects and publishes a mechanism $\mec$ satisfying Eq.~\eqref{eq:def-mechanism}.
\STATE Each agent $k$ decides whether to join the collaborative learning depending on Eq.~(\ref{eq:def-IR},\ref{eq:def-baseline-utility}). 
\STATE If joining, agent $k$ chooses a \textit{design} $w_{G_k} \in \R^{|G_K|}_+$ which maximizes her utility Eq.~\eqref{eq:utility-under-mechanism}. If the agent is strategic, then $w_{G_k}$ will correspond to the Nash equilibrium Eq.~\eqref{eq:def-nash-equilibrium}.
\STATE For all $i \in G_k$, she collects $w_i$ independent samples from $x_i$, incurring a cost $c^{(k)}$ per unit.
\STATE The agents commit all the collected data to a collaborative learning procedure coordinated by the principal. Based on this aggregated data, the principal computes the OLS estimator $\hat \theta$.
\STATE Then, to each agent $k$, the principal sends back a possibly degraded $\hat \theta^{(k)}$ in accordance with the published $\mec^{(k)}$.
\end{algorithmic}\label{alg:alg}
\end{algorithm}

\paragraph{Implementing the mechanism.} 
The mechanism needs to return a $\hat\theta^{(k)}$ to agent $k$ using data $\mec^{(k)}(w)$. This is equivalent to requiring that $\hat\theta^{(k)}$ is an unbiased estimator of the ground truth with covariance $\M\big(\mec^{(k)}(w)\big)^{-1}$. 

Another straightforward method of achieving this would be to run $K$ parallel federated learning algorithms. Each of these would train a model $\hat\theta^{(k)}$ for agent $k$ using only a subset of the data points as dictated by $\mec^{(k)}(w)$.

\begin{remark}[Computational burden]
    While implementing the full mechanism may seem computationally burdensome, note that we only incur this burden if $\mec^{(k)}(w) \neq w$. As we will see in the following sections, under equilibrium conditions we will always expect to see $\mec^{(k)}(w) = w$ and so no additional computation is required. The mechanism is merely a deterrent.
\end{remark}

Finally, our framework assumes that an agent only has access to the final output of the mechanism, but not to any intermediaries. This is important since if we are learning $\hat\theta$ using FL, the agents may utilize the intermediary estimates (which may be of better quality), instead of $\hat \theta^{(k)}$. This may be avoided by either assuming that the agents can be trusted to follow the protocol, or by appealing to security and cryptographic solutions.
\begin{remark}[Hiding intermediates]
The entire mechanism can be implemented in an encrypted/obfuscated software~\cite{barak2016hopes}, or in a trusted execution environment (TEE)~\cite{sabt2015trusted}. These solutions ensure that only the final output of the mechanism can be accessed and all intermediary computations remain hidden. Thus, the agents are prevented from cheating and follow our mechanism.
\end{remark} 

\subsection{Strategic behavior of agents}

In principal-agent experiment design problems, it is important to distinguish non-strategic agents and strategic agents. We say an agent $k$ is \emph{non-strategic} if she makes decisions solely depending on her design space $\X_k$. It is clear that to maximize the worst case utility, a non-strategic agent $k$ should simply optimize the single-agent utility Eq.~\eqref{eq:single_agent_utility} without considering contributions from the other agents, i.e., assuming $w_{G^c_k} = 0$.

However, the more theoretically interesting and practically relevant scenario arises when we consider strategic agents.  We say an agent $k$ is \emph{strategic} if she make decisions on the design $w_{G_k}$ depending on the decisions of other agents, knowing the design spaces $\X_{j}$ and the costs $c^{(j)}$ for all $j \in [K]$. 
We characterize the behaviors of strategic agents through the following definition.

\begin{definition}[Strategic responses]
We say the designs $w^\ast = (w_{G_1}^\ast,w_{G_2}^\ast,\dots,w_{G_K}^\ast)$ is a \emph{strategic response} to the mechanism $\mec = (\mec^{(1)},\mec^{(2)},\dots,\mec^{(k)})$ if:
\begin{itemize}[noitemsep, topsep=0pt, leftmargin=*]
\item Individual Rationality: For any $k \in [k]$, if $\sum_{i \in G_k}w^*_{i} > 0$ then
\begin{equation}\label{eq:def-IR}
    \big(u^{(k)} \circ \mec^{(k)}\big)(w^\ast) \geq v^{(k)}_\ast,
\end{equation}
where $v^{(k)}_\ast$ is the maximum possible utility agent $k$ can achieve if she opts out of the collaborated learning and trains a model using her own data:
\begin{equation}\label{eq:def-baseline-utility}
    v^{(k)}_\ast := \max_{w_{G_k}} f^{(k)}\Big( \big(\sum_{i \in G_k} w_i \cdot (A^{(k)})^\top x_ix_i^\top A^{(k)}\big)^{-1}\Big) - c^{(k)} \sum_{i \in G_k} w_i.
\end{equation}
\item Pure Nash Equilibrium: $(w_{G_1}^\ast,w_{G_2}^\ast,\dots,w_{G_K}^\ast)$ is the pure Nash equilibrium of the game defined by concave utilities $(u^{(k)} \circ \mec^{(k)})_{k \in [K]}$ and actions $(w_{G_k})_{k \in [K]}$.  That is, it satisfies 
\begin{equation}\label{eq:def-nash-equilibrium}
    \big(u^{(k)} \circ \mec^{(k)}\big)(w^\ast) \geq 
    \big(u^{(k)} \circ \mec^{(k)}\big)( w_{G_k}, w_{G_{k}^c}^\ast) ,  ~\forall ~w_{G_k} \in \R^{|G_k|}_+\,,
\end{equation}
for all $k \in [K]$.  Here, $( w_{G_k}, w_{G_{k}^c}^\ast)$ denotes concatenation: $(w_{G_1}^\ast,\dots,w_{G_{k-1}}^\ast, w_{G_k}, w_{G_{k+1}}^\ast, \dots,w_{G_K}^\ast)$.
\end{itemize}
\end{definition}

The first condition indicates that an agent will never choose an action that results in a worse outcome than $v^*$, the status quo that agent $k$ can obtain no matter she takes part in the collaboration learning or not. This constraint reflects \textit{ex post} individual rationality that ensures each agent in a collaborative learning setting achieves a minimum level of utility after the learning process is completed. The second condition asserts that an agent can not obtain higher utility by unilaterally changing her action. Thus, the strategic response $w^\ast$ represents a stable fixed point to the game from which no agent has an incentive to deviate from their chosen action. Further, if there is a unique Nash equilibrium, then this represents the only solution rational agents will play. 
The following result confirms the existence of pure Nash equilibrium for a range of optimality criteria and mechanisms.
\begin{proposition}\label{cla:existence_ne}
If for $\forall k \in [K]$, $f^{(k)}$ is differentiable that satisfies $\underset{\lambda \to +\infty}{\lim \sup} f^{(k)}\left((\lambda M)^{-1} \right)/(c \lambda) \leq 0$ for any $c > 0$ and p.s.d. matrix $M$, and $\mec^{(k)}$ is differentiable such that $(u^{(k)} \circ \mec^{(k)})(\cdot,w_{G_k^c})$ has unique maximizer or $\equiv -\infty$ for any fixed $w_{G_k^c}$. Then there exists pure Nash equilibrium of the game defined by utilities $(u^{(k)} \circ \mec^{(k)})_{k \in [K]}$ and actions $(w_{G_k})_{k \in [K]}$.
\end{proposition}
We present an additional sufficiency result for the existence of pure Nash equilibrium in Proposition~\ref{prop:existence_ne_2}, wherein the required conditions can be easily verified. Notably, the commonly employed optimality criteria listed in Section~\ref{sec:experiment_uncertainty} and the proposed mechanism outlined in Section~\ref{sec:data_max} satisfy all the required conditions stated in Proposition~\ref{cla:existence_ne} and Proposition~\ref{prop:existence_ne_2}. 
In the rest of the paper, we will analyze and design mechanisms with unique and desirable Nash equilibria.

\section{Efficiency and Free-Riding in Standard Federated Learning}\label{sec:fed_efficient}

The first question of interest in principal-agent experiment design is the efficiency of the mechanism. By classic optimal experiment design, a design measure $w = (w_{G_1},\dots,w_{G_K})$ is efficient for optimality criterion $f$ if $w$ is proportional to the optimal design measure $\pi^\ast = \arg \max f(\M(\pi)^{-1})$, for $\pi \in \Delta(\X)$.

\subsection{Incentive compatible efficiency}

% The 
In this section, we explore the conditions under which the standard federated learning mechanism which always sets $\mec^{(k)}_{\fed}(w) = w$ is efficient. We establish that the D-criterion is the only criterion among standard criteria for which the federated learning mechanism is efficient.

\begin{definition}[Incentive-compatibly efficient]
A mechanism $\mec$ is \emph{incentive-compatibly efficient} for a criterion $f$, if for any choice of design spaces $(\X_k)_{k\in [K]}$, all strategic responses $w^\ast$ are efficient designs for criterion $f$ and satisfy $w^\ast \propto \pi^\ast$.    
\end{definition}

\begin{proposition}\label{prop:fed_efficient_allocation}
Suppose $c^{(1)} = \cdots = c^{(K)} = c \in \R_+$. Then, among all optimality criteria, the federated learning mechanism ($\mec^{(k)}_{\fed}(w) = w$) is efficient only for the D-criterion. More precisely, 
\begin{enumerate}%[noitemsep, topsep=0pt]
    \item When all agents $k$ use criterion $f^{(k)}_D$, the agent's strategic response is the design given by $(\frac{d}{c} \cdot \pi^\ast_{G_k})_{k \in [K]}$, where $\pi^\ast \in \arg \max_{\pi \in \Delta(\X)} f_D(\M(\pi)^{-1})$. 
    \item For every other standard criteria (E, A, V, or G), there exists a design space $\X$ such that federated learning mechanism is not efficient.
\end{enumerate}
\end{proposition}

When each agent incurs the same marginal cost for sampling data, differences in efficiency can be attributed to the agents' data-generation capacities rather than variations in data acquisition costs. 
This setup allows for a fair comparison among agents and serves as the natural framework for studying efficiency. 
Our result implies that D-criterion is the only criterion that aligns the interest of each agent with the statistical efficiency of the multi-agent system. Therefore, it is the most suitable for experiment design problems involving multiple agents. 

\begin{remark}[Efficiency of D-optimality]
    That D-optimality uniquely satisfies incentive-compatible efficiency is remarkable. Numerous reviews and textbooks compare and contrast the different criteria but fail to identify a single best one~\cite{chaloner1995bayesian,fedorov1997model,pukelsheim2006optimal,atkinson2007optimum,goos2011optimal}. In fact, the popularity of D-optimality stemmed from its perceived equivalence to G-optimality, while being easier to optimize. The multi-agent perspective provides a novel lens with which to distinguish them and recommend the D-criterion over the rest. However, a note of caution is warranted---these results hold with our specific linear cost model. With different cost functions, it is possible that the conclusions differ.
\end{remark}

\begin{remark}[Invariance to linear transformation]
We remark that another desirable property of the D-criterion is the invariance to linear transformation. 
Our framework assumes prior knowledge of $A^{(k)}$ or requires its truthful reporting. 
This flexibility may potentially affect the incentives of strategic agents and raise concerns regarding statistical efficiency. Moreover, in practical collaborative learning scenarios~\cite{xu2023fair}, it is common for agents to deploy the model on a target domain that differs from the source design space $\X_k$. In all such cases, the local information matrix in Eq.~\ref{eq:def-local-information} may undergo a domain shift to
\begin{align}\label{eq:domain_shift}
    \M^{(k)}(w) = \left(({A^{(k)}}T)^\top\M(w)^{\dagger}A^{(k)}T\right)^{-1}
\end{align}
for some linear transformation $T \in \R^{r_k \times r_k}$. 
However, since the D-criterion only changes by a constant additive factor after linear transformation, the incentive of agent $k$ will be unaffected after applying a linear transformation of $A^{(k)}$. This indicates that the strategic responses are independent of the domain shift in the form of Eq.~\eqref{eq:domain_shift}, and thus incentive compatible efficiency is preserved. This property can be viewed as multi-dimension version of scale invariance in the Nash bargaining solution \cite{nash1950bargaining}. 
\end{remark}

The above results leave the question of efficiency under heterogeneous costs. The standard federated learning does not suffice any longer, and we instead require non-trivial mechanisms. We defer a discussion of this issue to Appendix~\ref{sec:efficient_allocation}.

\subsection{Free-riding behavior}

Although the federated learning mechanism $\mec_{\fed}$ achieves efficiency in the multi-agent system, it can lead to unfair Nash equilibria in which some agents contribute many fewer data points than others. This phenomenon, known as free-riding, is highly undesirable in federated learning \cite{baumol2004welfare,choi2019contributors,sarikaya2019motivating,lin2019free,ding2020incentive,sim2020collaborative,xu2021gradient}. We illustrate two possible cases where agents typically gain more utility by free-riding.
\begin{example}[Free-riding due to data diversity]
Consider a principal-agent experiment design problem where one agent possesses a data set with high diversity, such that her design space covers the design space of the other agent. In such cases, it can be demonstrated that the second agent will engage in free-riding behavior at a pure Nash equilibrium. We establish the following result to formalize this scenario.
\begin{proposition}\label{prop:free_riding_diversity}
Suppose agent $k$'s design space $\X_k$ and agent $l$'s design space $\X_l$ satisfy the condition 
$\{x_ix_i^\top: i \in G_l\} \subset \left\{\sum_{i \in G_k}\alpha_i x_ix_i^\top: \alpha \in \R^{|G_k|}_+, \sum_{i \in G_k} \alpha_i < 1, \right\}$.  Then, in any pure Nash equilibrium, $w_{G_l} = 0$.
\end{proposition}
\end{example}

\begin{example}[Free-riding due to cost heterogeneity]
Consider another scenario where strategic agents with higher marginal costs may engage in free-riding behavior. Intuitively, in equilibrium, an agent with a lower marginal cost experiences a higher marginal increase in utility by sampling more data. If another agent possesses equal experimental capacity but at a higher cost, she is expected to engage in free-riding at a pure Nash equilibrium. 
More precisely, we have the following result.
\begin{proposition}\label{prop:free_riding_cost}
Suppose $\X_k = \X_l$ for some $k \neq l$ and $c^{(k)} < c^{(l)}$. Then, in any pure Nash equilibrium, we have $w_{G_l} = 0$.
\end{proposition}
\end{example}
These examples highlight situations where agents have incentives to free-ride due to factors such as data diversity or cost disparities. Such behaviors can undermine the fairness and collaboration within the multi-agent system. In the subsequent sections, we delve into the analysis of free-riding behaviors and propose mechanisms to mitigate these issues.

\section{Information Maximization}\label{sec:data_max}

In this section, we address the second question posed in the introduction. 
Adopting information-theoretic concepts~\cite{kay1993fundamentals,cover1999elements}, the total information is quantified by $\log\det \M(w)$, which is proportional to the negative of differential entropy of $\hat{\theta}$ for standard Gaussian errors. 
Building upon Proposition~\ref{prop:fed_efficient_allocation}, we assume that every agent $k$ uses the D-criterion $f^{(k)}_D$ throughout this section. It is worth noting that this choice is also compatible with our information-theoretic considerations since maximizing the D-criterion is equivalent to minimizing the differential entropy of $\hat{\theta}^{(k)}$.

\subsection{Maximum possible information by rational agents}

To achieve information maximization, we need to first understand what is the maximum information that could be possibly generated by strategic agents. 
We have the following result on the maximum achievable information.
\begin{proposition}[Maximum information]\label{prop:max_possible_information}
Define 
\begin{align}\label{eq:max_object_design_linear}
w_{\max} := \arg \max_{w \in \R^n_+} \log \det \M(w),
    ~\text{s.t.}~ u^{(k)}(w) \geq v^{(k)}_\ast .
\end{align}
Then, for any mechanism $\mec$ and any strategic response $\wt w$ under $\mec$, we have $\log \det \M(\wt w) \leq \log \det \M(w_{\max})$\,.
\end{proposition}
 The definition of $w_{\max}$ above maximizes over all possible contributions $w$ by the agents which satisfy individual rationality assuming full data sharing. This is a superset of strategic responses by agents since agents will additionally require a Nash equilibrium. Thus, $\log \det \M(w_{\max})$ represents the maximum possible information achievable in a setting where the agents are non-strategic and simply follow the principal's prescribed strategy.
 
 However, in reality we need to account for the agency of the agents. We first investigate if the common federated learning mechanism $\mec_{\fed}$ suffices. Unfortunately, the following proposition shows that, in general, $w_{\max}$ cannot be achieved under $\mec_{\fed}$, as there exists at least one agent who can achieve higher utility by contributing fewer samples.

\begin{proposition}[Free-riding under federated learning ]\label{prop:free_riding}
Unless $\sum_{k=1}^K r_k = d$, $w_{\max}$ is not the Nash equilibrium of the utility functions $\big(\big(u^{(k)}\circ \mec_{\fed}^{(k)}\big)\big)_{k \in [K]}$. More precisely, there exists $k \in [K]$ and $\wt w_{G_k}$ such that $\wt w_i \leq w_{\max,i},\forall i \in G_k$; $\wt w_i < w_{\max,i},\exists i \in G_k$, and
\begin{align*}
    \big(u^{(k)}\circ \mec^{(k)}_{\fed}\big)\big((\wt w_{G_k}, w_{\max,G_k^c})\big) > \big(u^{(k)}\circ \mec^{(k)}_{\fed}\big)( w_{\max})
\end{align*}
where $(\wt w_{G_k}, w_{\max,G_k^c})$ denotes the concatenation of $\wt w_{G_k}$ and $w_{\max,G_k^c}$.
\end{proposition}

The condition $\sum_{k=1}^K r_k = d$ requires that the covariates in $\X_k$ for each agent $k$ form independent subspaces, and each agent cannot benefit from the data from other agents. Thus, this condition is unlikely to be encountered in the study of collaborative learning.

\subsection{Information-maximizing mechanism}

Motivated by the Proposition~\ref{prop:free_riding_diversity}, Proposition~\ref{prop:free_riding_cost} and Proposition~\ref{prop:free_riding}, we design mechanisms $(\mec^{(k)}_{\max})_{k \in [K]}$ to incentive agents to contribute $w_{\max}$ amount of data. Let $\mec^{(k)}_{\max}$ simply scale the design by a constant $\gamma_k \leq 1$:
\begin{align}\label{eq:mechanism_max_data}
    \mec^{(k)}_{\max}(w) := \gamma_k w \,,\quad \text{ for } \gamma_k^{-1} := \exp\big(\tfrac{c^{(k)}}{r_k}\cdot\sum_{i \in G_k}(w_{\max,i} - w_i)_+ \big)\,.
\end{align}
In this mechanism, agents are penalized for contributing less data than required for information maximization ($w_{\max}$).
The $k$-th agent's utility $\big(u^{(k)}\circ \mec_{\max}^{(k)}\big)(w)$ is then given by
\begin{align*}%\label{eq:utility_incentivized_cost_max}
    -\log\det \big((A^{(k)})^\top \M(w)^{\dagger} A^{(k)}\big) - c^{(k)} \sum_{i \in G_k} w_i - c^{(k)}\sum_{i \in G_k}\left(w_{\max,i} - w_i\right)_+.
\end{align*}

We have the following proposition that establishes information maximization as the unique strategic response of the information mechanism $\mec_{\max}$.

\begin{proposition}[Information maximization]\label{prop:data_max}
The information-maximization design $(w_{\max,G_k})_{k \in [K]}$ in Eq.~\eqref{eq:max_object_design_linear} is the unique strategic response of the agents to the information mechanism $\mec_{\max}$ in Eq.~\eqref{eq:mechanism_max_data}.
\end{proposition}

\begin{remark}[Simplicity of the penalization scheme]
Proposition~\ref{prop:data_max} shows that it is indeed possible for a mechanism to achieve information $\log \det \M(w_{\max})$. This is surprising since the definition of $w_{\max}$ assumed that agents would simply follow the principal's prescribed design and not strategize. It is further remarkable that uniformly penalizing by a scalar as in Eq.~\eqref{eq:mechanism_max_data} suffices, rather than a more bespoke mechanism penalizing different direction differently. Intuitively, this simplicity is a result of the D-optimality criterion already satisfying incentive-compatible efficiency.
\end{remark}

Proposition~\ref{prop:data_max} establishes that information-maximization design, represented by the $(w_{\max,G_k})_{k \in [K]}$, is the unique strategic response to the mechanism $\mec_{\max}$. This thus addresses the second question posed in the introduction and provides a way for a federated learning community to maximize data creation from multiple autonomous parties and generate positive societal impact~\cite{graham2015disparities,albain2009racial,hudson2020rights,chediak2020sharing,zhan2021survey,shi2021survey}. 
In Appendix~\ref{sec:further_discussion}, we present a further discussion of this impact, examining notions of fairness to ensure that none of the participating agents will be taken advantage of, and investigating the inefficiency of the system by quantifying the degradation in the collected information caused by the self-interested participants. 

\begin{remark}[Implementation via early stopping]
    The information-maximizing mechanism (Eq.~\ref{eq:mechanism_max_data}) is quite simple: it scales the design measure by a scalar $\gamma_k \in [0,1]$. This corresponds to randomly sub-sampling a $\gamma_k$ fraction of the data to train $\hat\theta^{(k)}$. Instead of implementing this via $K$ parallel sub-samplings and federated learning runs, a more convenient approximation may be achieved using early stopping. Intuitively, early stopping also effectively subsamples data. During training of the global model $\hat\theta$ for a total of $T$ rounds, the model at the $\gamma_k T$ round is returned to agent $k$ as its $\hat\theta^{(k)}$.
\end{remark}

\subsection{Simulation}
\label{sec:simulation}

\begin{figure}[H]\label{fig:simulation}
\centering
\begin{tabular}{cccc} 

\includegraphics[width=0.22\textwidth]{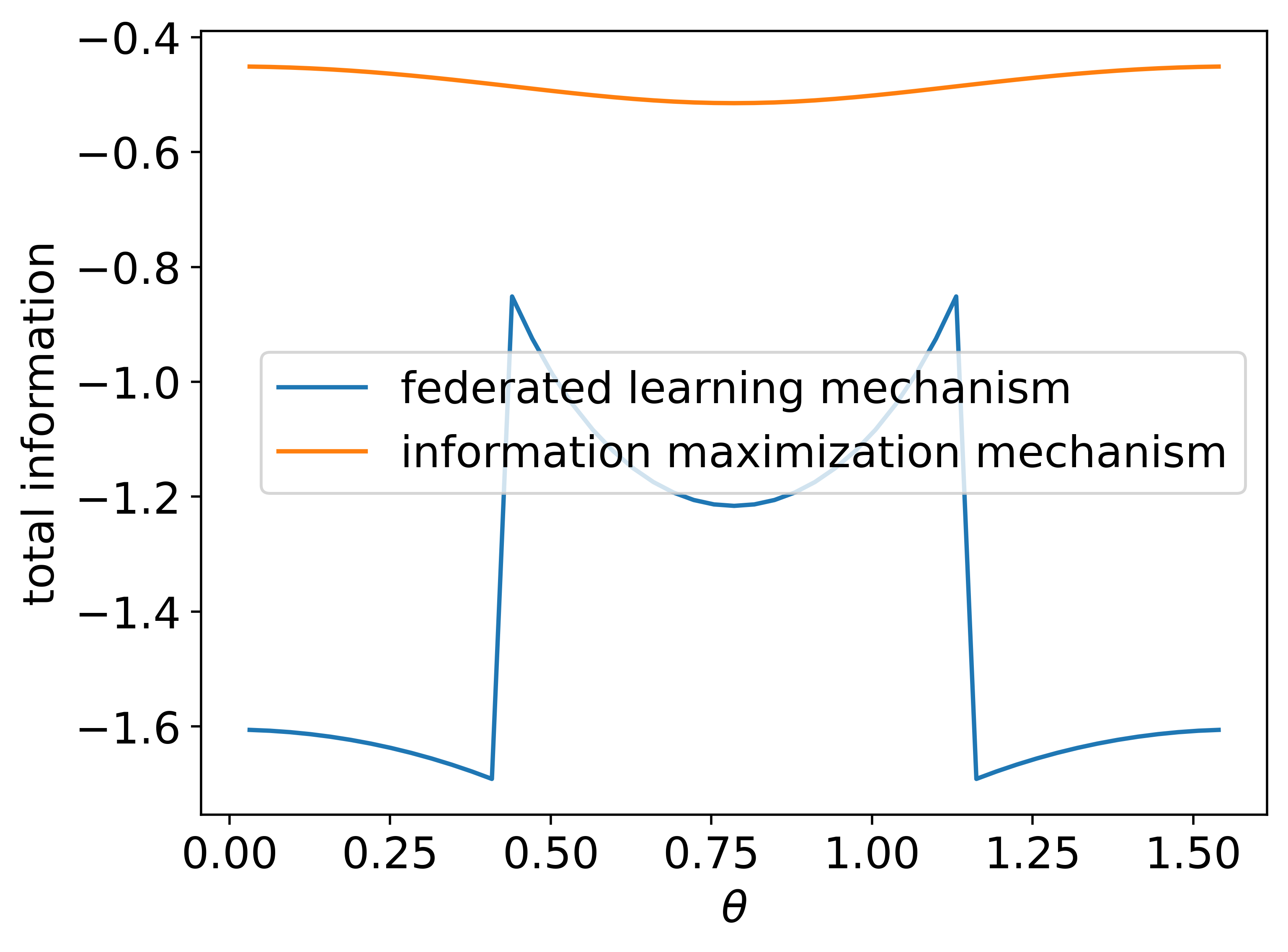} &
\includegraphics[width=0.22\textwidth]{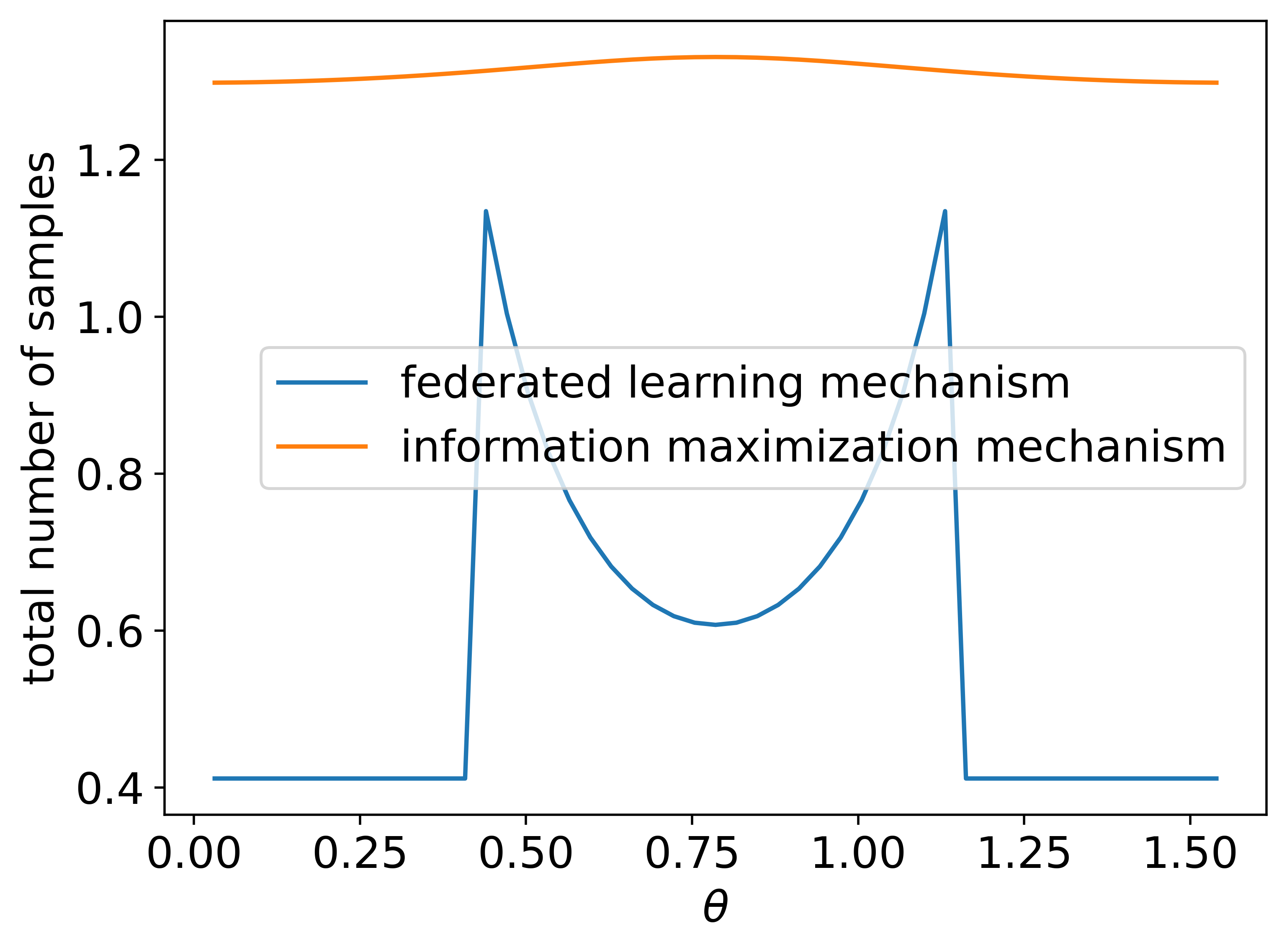} &
\includegraphics[width=0.22\textwidth]{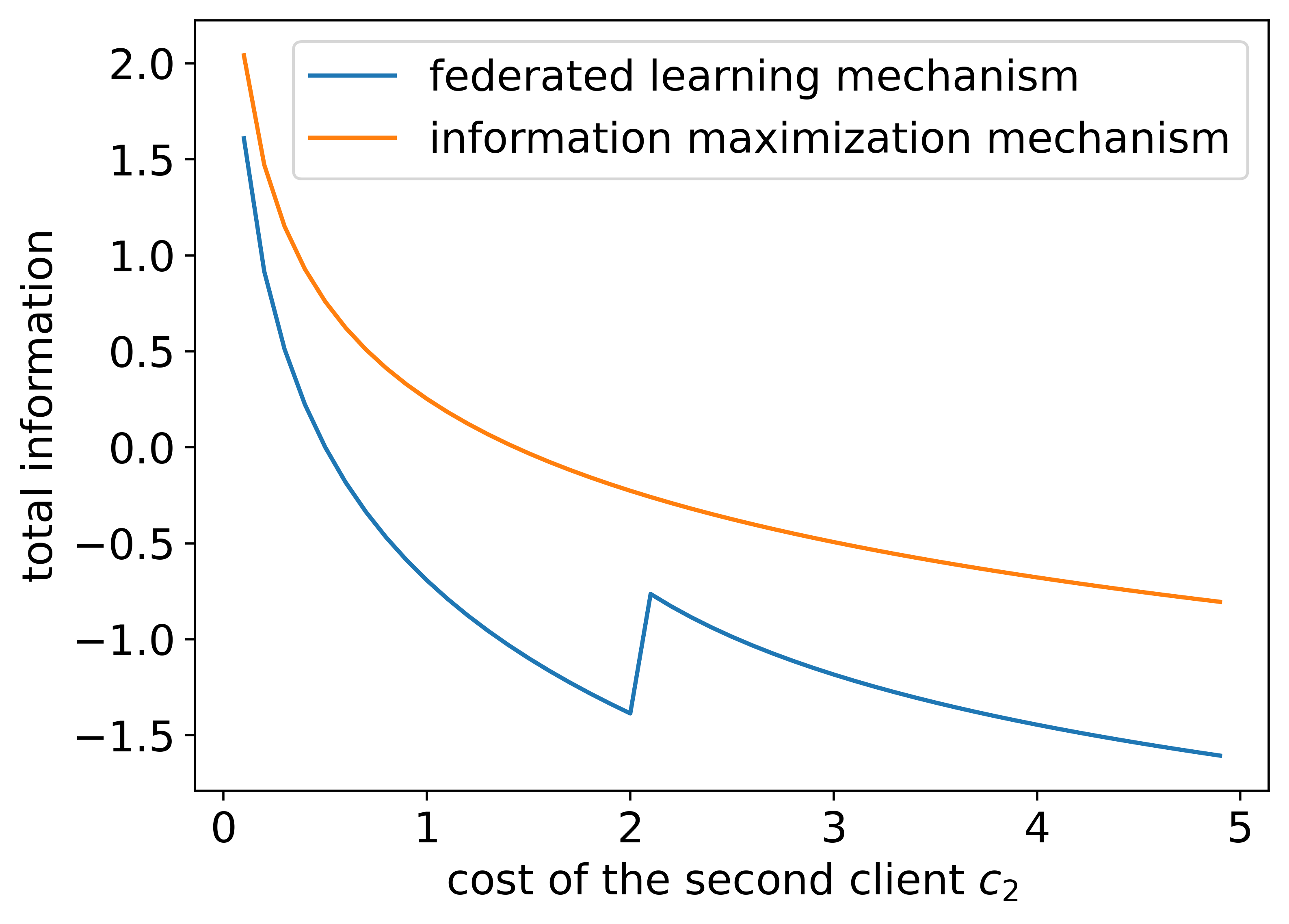}&
\includegraphics[width=0.22\textwidth]{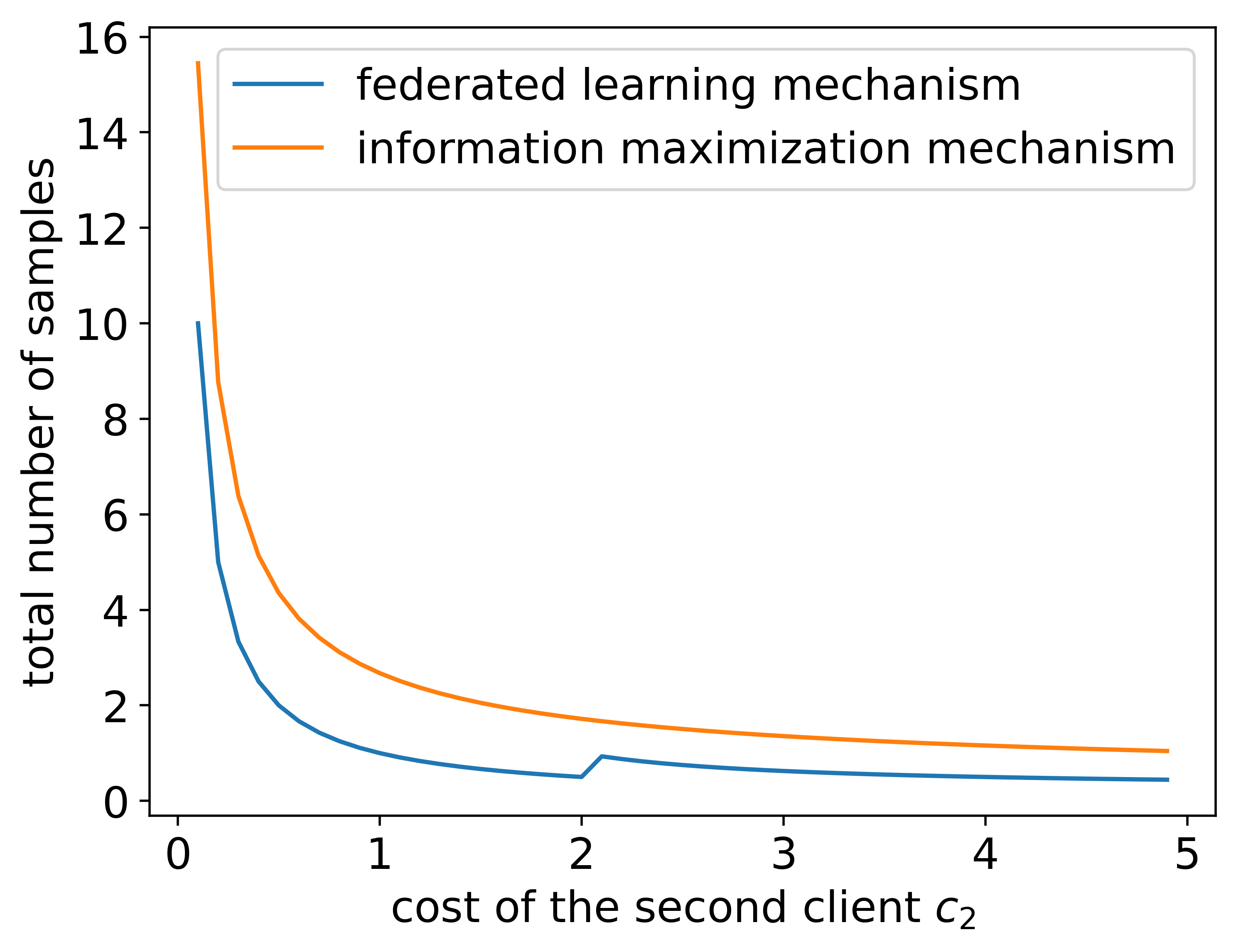}\\

\scriptsize{$(a)$} & \scriptsize{$(b)$}%\\
&\scriptsize{$(c)$} & \scriptsize{$(d)$}
\end{tabular}

\caption{Comparison between federated learning mechanism and information-maximization mechanism for different $\theta$ and $c^{(2)}/c^{(1)}$, where (a) and (b) are plotted with fixed $c^{(1)} = 2, c^{(2)} = 3$, (c) and (d) are plotted with fixed $c^{(1)} = 2, \theta = \pi/4$. (a) total information as function of $\theta$; (b) $w_1+w_2+w_3$ as function of $\theta$; (c) total information as function of $c^{(2)}$; (d) $w_1+w_2+w_3$ as function of $c^{(2)}$. In all cases our mechanism elicits significantly more data contributions.}
\end{figure}

In this section, we conduct experiments using a toy example to illustrate the impact of the information-maximization mechanism. The example involves the following design space: $\X = \{x_1 = (\cos \theta, \sin \theta)^\top, x_2 = (1,0)^\top,x_3 = (0,1)^\top\}$, and groups $G_1 = \{1\}$ and $G_2 = \{2,3\}$. Notably, $x_2$ and $x_3$ are orthogonal, while the parameter $\theta$ governs the degree of complementarity between $x_1$ and $x_2$. 
We investigate the strategic behaviors of the federated learning and information-maximization mechanism by varying $\theta$ and $c^{(2)}/c^{(1)}$. Figure~\ref{fig:simulation} presents the total data contribution $w_1, w_2, w_3$, as well as the total information, for different values of $c^{(2)}$ and $\theta$ fixing $c^{(1)}$. These visualizations demonstrate how cost heterogeneity and data diversity influence the strategic response across various mechanisms. 
Remarkably, the information-maximization mechanism $\mec_{\max}$ yields improved data contribution and information while exhibiting a more stable behavior.

\section{Conclusions}
We have formulated the problem of principal-agent experiment design to capture the game-theoretic tensions between the principal and strategic agents in collaborative learning. 
We showed that under standard federated learning, strategic agents will adopt the optimal design strategy if and only if the D-optimality criterion is used. Additionally, we have highlighted that strategic agents often exhibit free-riding behavior, driven by factors such as data diversity and cost heterogeneity. This observation has motivated us to develop a mechanism that incentivizes strategic agents to maximize the overall information. The proposed mechanism has significant societal implications as it promotes autonomy and equity in domains such as clinical trials, collaborative cancer research, and other medical and scientific studies where optimal experiment design is heavily used. It may also be independent interest in promoting collaboration and fairness in active learning problems.

Our results come with some limitations, while opening new avenues for future research. Firstly, our framework does not analyze concrete algorithms with realistic considerations such as unknown design spaces $\X_k$. Overcoming this is an important direction of future work. Furthermore, it would be intriguing to generalize our results to mixed effect models or nonlinear models, which would broaden the scope of our analysis and uncover additional nuances in the principal-agent experiment design problem. Finally, our theoretical analysis uses a somewhat stylized model for the behavior of agents. Translating the insights gained in our work to the real world is challenging but necessary.

\bibliography{ref}
\ifdefined\isarxiv
\bibliographystyle{IEEE}
\else
\bibliographystyle{IEEE}
\fi

\newpage
\appendix

\onecolumn

\section{Optimality Criteria for Self-interested Agents}\label{sec:local_criterion}

\begin{align*}%\label{eq:optimality_criteria}
    \text{D-criterion:}\quad &~ f^{(k)}_D(M^{-1}) = \log \det M \notag\\
    \text{G-criterion:}\quad &~ f^{(k)}_G(M^{-1}) = - \max_{x \in \X_k} \left[({A^{(k)}}^\top x)^\top M^{-1} ({A^{(k)}}^\top x)\right] \notag\\
    \text{E-criterion:}\quad &~ f^{(k)}_E(M^{-1}) = -\|M^{-1}\|_2 \\
    \text{A-criterion:}\quad &~ f^{(k)}_A(M^{-1}) = -\tr \left[M^{-1}\right] \notag\\
    \mathrm{V}_{p^{(k)}} \text{-criterion:}\quad &~ f^{(k)}_{V_{p^{(k)}}}(M^{-1}) = - \E_{x \sim p^{(k)}}\tr \left[({A^{(k)}}^\top x)^\top M^{-1}({A^{(k)}}^\top x)\right] \notag
\end{align*}
In the $\mathrm{V}_{p^{(k)}}$-criterion, $p^{(k)}$ is a distribution supported on $\X_k$. We take $p^{(k)}$ as uniform distribution over $\X_k$ when we refer to V-criterion without specifying $p^{(k)}$. Indeed, we can simulate the $\mathrm{V}_{p^{(k)}}$-criterion by repeating the elements in $\X_k$ according to $p^{(k)}$ and using V-criterion in the augmented design space. $\mathrm{V}_{p^{(k)}}$-criterion can be used to model self-interested agent $k$ that is interested in minimizing the mean squared error under distribution $p^{(k)}$.

\section{Further Discussion of Information-Maximization Mechanisms}\label{sec:further_discussion}

\subsection{Fairness}

The $\mec_{\max}$ mechanism ensures that the principal obtains the maximum possible information while preventing any agent from free-riding. Meanwhile, it is also crucial to examine notions of fairness to ensure that none of the participating agents are exploited. In order to address this, we analyze the utility of agents under the mechanism $\mec_{\max}$ and present the following result.
\begin{corollary}[Incentive Compatibility]\label{cor:incentive_compatibility}
Under mechanism $\mec_{\max}$, the strategic response $w_{\max}$ satisfies $\left(u^{(k)}\circ \mec_{\max}^{(k)}\right)\left(w_{\max}\right) = v^{(k)}_*$ for all $k \in [K]$.
\end{corollary}

The above corollary is straightforward from the optimization problem in Eq.~\eqref{eq:max_object_design_linear}. Nevertheless, it carries two important implications. 
First, this corollary implies that the utility obtained by agent $k$ through strategic participation in the collaborative learning, given by $\left(u^{(k)}\circ \mec_{\max}^{(k)}\right)(w_{\max})$, is equal to the maximum utility $v^{(k)}_*$ that the agent can obtain by training individually. Therefore, all participating agents benefit equally from the collaborative learning process. 
In fact, the surplus generated by agents is directed towards enhancing the value of the statistical model, ultimately benefiting the social welfare. This highlights the equitable distribution of benefits and the collective progress achieved through collaboration. 
Secondly, Corollary~\ref{cor:incentive_compatibility} implies that the utility of agent $k$ under the mechanism depends solely on the resources and capacities of agent $k$ itself, represented by $\X_k, A^{(k)}, f^{(k)}, c^{(k)}$, and is independent of other agents. Consequently, any improvements or innovations made by agent $k$ to enhance experimental conditions or reduce marginal costs will be fully exploited within the mechanism. This incentivizes participating agents to enhance their own capacities and resources, promoting an environment of continuous improvement. Thus, the mechanism $\mec_{\max}$ exhibits incentive compatibility, fostering agents' motivation to optimize their contributions.

The issue of fairness in the principal-agent experiment design problem is particularly relevant in the exchangeable data setting, where all data points have the same value \cite{karimireddy2022mechanisms}. In this scenario, there are no inherent distinctions between the resources and targets of different agents, therefore demanding the mechanism to avoid introducing extrinsic unfairness among the agents. Fortunately, our proposed mechanism satisfies a monotonic notion of fairness.

\begin{proposition}[Fairness under exchangeable data regime]\label{prop:max_fairness}
In the exchangeable data regime (i.e., $X_i$'s are the same), the information maximization mechanism $\mec_{\max}$ is fair in the sense that any strategic response $\bar w = (\bar w_{G_1},\bar w_{G_2},\dots,\bar w_{G_K})$ satisfies that for all $k,k' \in [K]$
\begin{align*}
    \left(u^{(k)}\circ \mec_{\max}^{(k)}\right)\left(\bar w\right) \geq \left(u^{(k')}\circ \mec_{\max}^{(k')}\right)\left(\bar w\right)  \implies \|\bar w_{G_k} \|_1 \geq \|\bar w_{G_k'} \|_1.
\end{align*}
\end{proposition}

This proposition states that in the exchangeable data regime, an agent must contribute more data in order to achieve a higher utility, which aligns with existing notions of fairness in the federated learning literature \cite{yu2020fairness, donahue2021optimality, donahue2023fairness}. When the data points are not exchangeable, fairness becomes more challenging to define due to the inherent heterogeneity of learning targets and resources. We leave the discussions regarding fairness in such scenarios the subject of future research. 

\subsection{Price of anarchy}

In this section we discuss price of anarchy \cite{koutsoupias1999worst} of the information maximization mechanism $\mec_{\max}$.

\begin{definition}[Price of Anarchy]
We define the social good as
\begin{align*}
    \sg(w) = \sum_{k=1}^K \left(u^{(k)}\circ \mec_{\max}^{(k)}\right)(w).
\end{align*}
and price of anarchy by the ratio between the maximal social good and the social good at strategic response, i.e.
\begin{align*}
    \poa := \max_{w} \frac{\sg(w)}{\sg(w_{\max})}.
\end{align*}
\end{definition}

Price of anarchy measures the inefficiency and suboptimality resulting from strategic behaviors in principal-agent experiment design. The numerator is the optimal `centralized' social good that can be achieved from the strategy spaces, and the denominator captures the social welfare obtained under selfish behaviors of each agent. To characterize price of anarchy of $\mec_{\max}$, we introduce the following concept.

\begin{definition}[Benefit from collaboration]
Define the Benefit from Collaboration of client $k$ as
\begin{align*}
    \Delta^{(k)} = &~ \max_{\pi \in \Delta([n])} - \log \det \left((A^{(k)})^\top \left(\sum_{i =1}^n \pi_ix_ix_i^\top \right)^{-1} A^{(k)}\right)\\
    &~ - \max_{\pi \in \Delta(G_k)} \log \det \left((A^{(k)})^\top \left(\sum_{i \in G_k} \pi_ix_ix_i^\top \right) A^{(k)}\right).
\end{align*}
\end{definition}

Intuitively, $\Delta^{(k)}$ describes the maximum achievable increase of information for agent $k$ by joining the collaborative learning. 
We will show that the price of anarchy is bounded by the benefit from collaboration.

\begin{proposition}\label{prop:poa_max}
Define $k_0 = \arg \min_{k\in [K]}c^{(k)}$.
then $\poa$ can be upper bounded by 
\begin{align*}
    \frac{\sum_{k=1}^K  \Delta^{(k)}}{\sum_{k=1}^K \left( \theta^{(k)} + r_k\log \frac{r_k}{c^{(k)}}  - r_k\right)} + \frac{\sum_{k=1}^K \left(r_k\log \frac{c^{(k)}\sum_{k=1}^K r_k}{r_kc^{(k_0)}} - (c^{(k)} - c^{(k_0)}) \cdot \| w_{\max, G_k} \|_1\right)}{\sum_{k=1}^K \left( \theta^{(k)} + r_k\log \frac{r_k}{c^{(k)}}  - r_k\right)}+1.
\end{align*}
\end{proposition}

To interpret the bound, we notice that the first term, $\frac{\sum_{k=1}^K  \Delta^{(k)}}{\sum_{k=1}^K \left( \theta^{(k)} + r_k\log \frac{r_k}{c^{(k)}}  - r_k\right)}$, represents the price of anarchy resulting from data diversity. It captures the extent to which each agent, denoted by $k$, benefits from a more diverse collection of data points contributed by other agents, which has the potential to improve the low-rank model of agent $k$. The second term, $\frac{\sum_{k=1}^K \left(r_k\log \frac{c^{(k)}\sum_{k=1}^K r_k}{r_kc^{(k_0)}} - (c^{(k)} - c^{(k_0)}) \cdot \| w_{\max, G_k} \|_1\right)}{\sum_{k=1}^K \left( \theta^{(k)} + r_k\log \frac{r_k}{c^{(k)}}  - r_k\right)}$, captures the price of anarchy resulting from cost heterogeneity and shared representation. It accounts for the potential exploitation of lower costs by the system in a centralized setting and the benefits of utilizing data collected from design spaces of rank $r_k$ to improve the model across all rank $r_{k'}$ spaces for $k \in [K]$. Notice that the $\frac{\sum_{k=1}^K \left(- (c^{(k)} - c^{(k_0)}) \cdot \| w_{\max, G_k} \|_1\right)}{\sum_{k=1}^K \left( \theta^{(k)} + r_k\log \frac{r_k}{c^{(k)}}  - r_k\right)}$ is a negative term that demonstrates the cost heterogeneity mitigated by the information maximization mechanism $\mec_{\max}$.

\clearpage
\section{Examples}\label{sec:examples}
In this section, we present several illustrative examples that highlight the strategic behaviors of self-interested agents in different scenarios.

\begin{example}[Free riding]
Consider $\X=\{x_1 = (1,0,0)^\top, x_2 = (0,1,0)^\top, x_3 = (0,0,1)^\top, x_4 = (0,1,1)^\top\}$ and the index sets given by $G_i = \{i\}$ for $i = 1,2,3,4$, i.e. four agents each holding a rank-$1$ set of experiment condition. 
Now, let's assume that the cost for the agents are $c^{(1)} = c^{(2)} = c^{(3)} = c \leq 0.5 c^{(4)}$. 

In this setup, we can observe that the Nash equilibrium for the standard federated learning mechanism is achieved when $w_1 = w_2 = w_3 = \frac{1}{c}$ and $w_4 = 0$. However, in this Nash equilibrium, agent $x_4$ contributes nothing to the collaborative learning process while benefiting from the information provided by the second and the third agents. This behavior, where agents exploit the contributions of others without contributing themselves, is known as free riding in federated learning.

\end{example}

\begin{example}[Selfish allocation]
Consider $u,v > 0$ and $\X = \{x_{2i+1} = u \cdot e_{i+1} \in \R^n, x_{2i+2} = e_1 \in \R^n ~ (i = 1,2,\dots,n-1), x_{2n+1} = v \cdot e_1 \in \R^n\}$ and the index sets given by $G_i = \{2i+1,2i+2\} ~(i = 1,2,\dots,n-1), G_n = \{2n+1\}$. That is, there are $n$ agents; the first $n-1$ agents each holds a rank-$2$ set $\{u \cdot e_{i+1}, e_1\}$ where $e_1$ can be seen as a shared feature and $e_{i+1}$ can be seen as the unique feature; the $n$-th agent holds $\{v \cdot e_1\}$. 

In this setup, each agent has a distinct feature and a shared feature. The first $n-1$ agents may selfishly conduct experiments only on their unique feature ($u \cdot e_{i+1}$) while hoping that other agents would experiment on the shared feature ($e_1$). This results in a selfish allocation of experiments, which can be highly inefficient.

For example, when $c^{(1)} = \cdots = c^{(n-1)}  \geq c^{(n)}/v^2$, it is clear that $w_{2i+1} = \frac{1}{c^{(i)}}~(i = 1,2,\dots,n), w_{2i+2} = 0 ~(i = 1,2,\dots,n-1)$ is a Nash equilibrium for the standard federated learning mechanism. In this Nash equilibrium, the first $n-1$ agents only experiment on $u \cdot e_{i+1}$ and in the end only the $n$-th agent samples from $x_{2n+1} = v \cdot e_1$. However, when $v \ll 1$, $x_{2i+2}$ gives more information and lies in the support of optimal experiment design instead of $x_{2n+1}$. Therefore, the presented strategic response is highly efficient. This is an example of selfish allocation in federated learning.
\end{example}

\begin{example}[Case study of substitutable, orthogonal, and complimentary data]
We continue the discussion of Section~\ref{sec:simulation}. 
We study the Nash equilibrium under federated learning mechanism with varying $\theta$ and $c^{(2)}/c^{(1)}$. 
By direct computation, the strategic response is given by
\begin{align*}
    w_1 = &~ \begin{cases}0, &~ c^{(2)} < c^{(1)}\\
    \frac{1}{c^{(1)}}, &~ c^{(1)} + c^{(2)}(\sin^2 \theta - \cos^2 \theta) < 0\\
    \frac{1}{c^{(1)}}, &~ c^{(1)} - c^{(2)}(\sin^2 \theta - \cos^2 \theta) < 0\\
    \frac{c^{(2)} - c^{(1)}}{c^{(2)}c^{(1)} - (c^{(1)})^2 - (c^{(2)})^2 (\sin^2 \theta - \cos^2 \theta)^2} , &~ \text{else}
    \end{cases}\\
    w_2 = &~ \begin{cases}\frac{1}{c^{(2)}}, &~ c^{(2)} < c^{(1)}\\
    0, &~ c^{(1)} + c^{(2)}(\sin^2 \theta - \cos^2 \theta) < 0\\
    \frac{1}{c^{(2)}}, &~ c^{(1)} - c^{(2)}(\sin^2 \theta - \cos^2 \theta) < 0\\
    \frac{c^{(1)} + c^{(2)}(\sin^2 \theta - \cos^2 \theta)}{c^{(2)}c^{(1)} - (c^{(1)})^2 - (c^{(2)})^2 (\sin^2 \theta - \cos^2 \theta)^2} , &~ \text{else}.
    \end{cases}\\
    w_3 = &~ \begin{cases}\frac{1}{c^{(2)}}, &~ c^{(2)} < c^{(1)}\\
    \frac{1}{c^{(2)}}, &~ c^{(1)} + c^{(2)}(\sin^2 \theta - \cos^2 \theta) < 0\\
    0, &~ c^{(1)} - c^{(2)}(\sin^2 \theta - \cos^2 \theta) < 0\\
    \frac{c^{(1)} - c^{(2)}(\sin^2 \theta - \cos^2 \theta)}{c^{(2)}c^{(1)} - (c^{(1)})^2 - (c^{(2)})^2 (\sin^2 \theta - \cos^2 \theta)^2} , &~ \text{else}.
    \end{cases}.
\end{align*}
This leads to sub-optimal amount of total information compared to the information maximization regime. We show the data contribution $w_1,w_2,w_3$ and the total information for varying $c^{(2)}$ and $\theta$ in Figure~3 and Figure~4.

\begin{figure}[H]\label{fig:vary_c2}
\centering
\begin{tabular}{ccc} 
\includegraphics[width=0.4\textwidth]{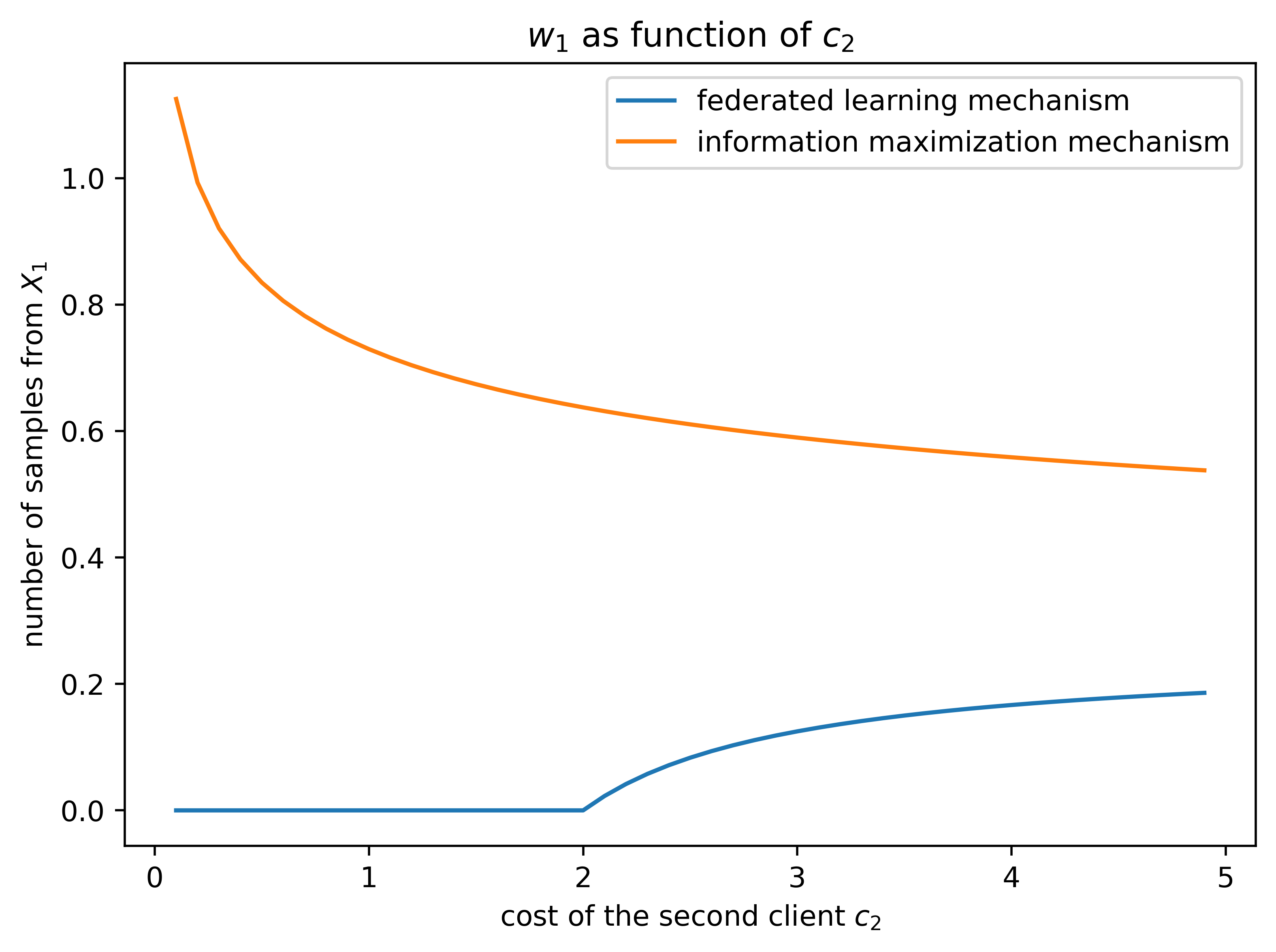} &
\includegraphics[width=0.4\textwidth]{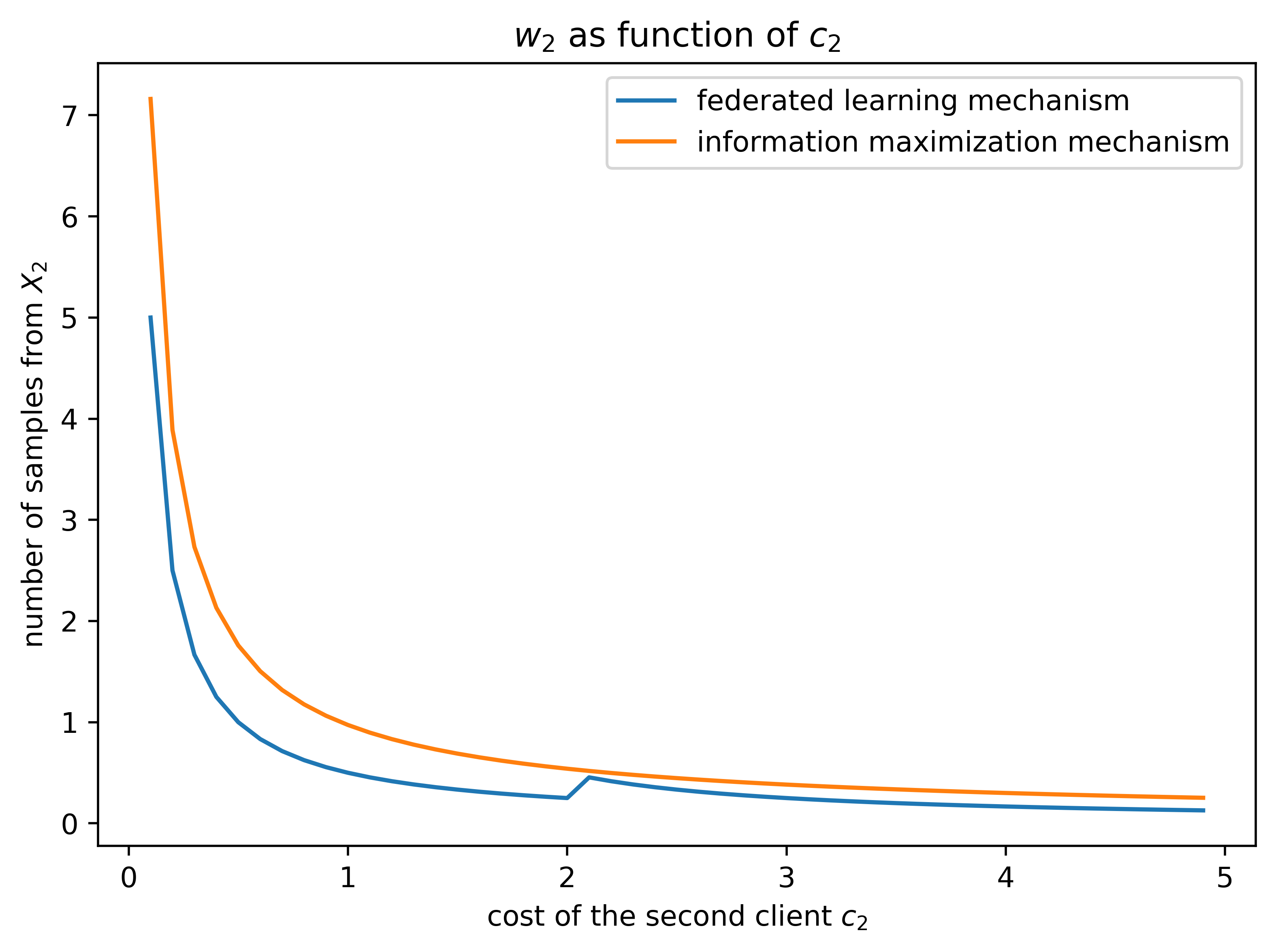}\\\vspace{-0.5em}
\scriptsize{$(a)$} & \scriptsize{$(b)$}\\
\\
\includegraphics[width=0.4\textwidth]{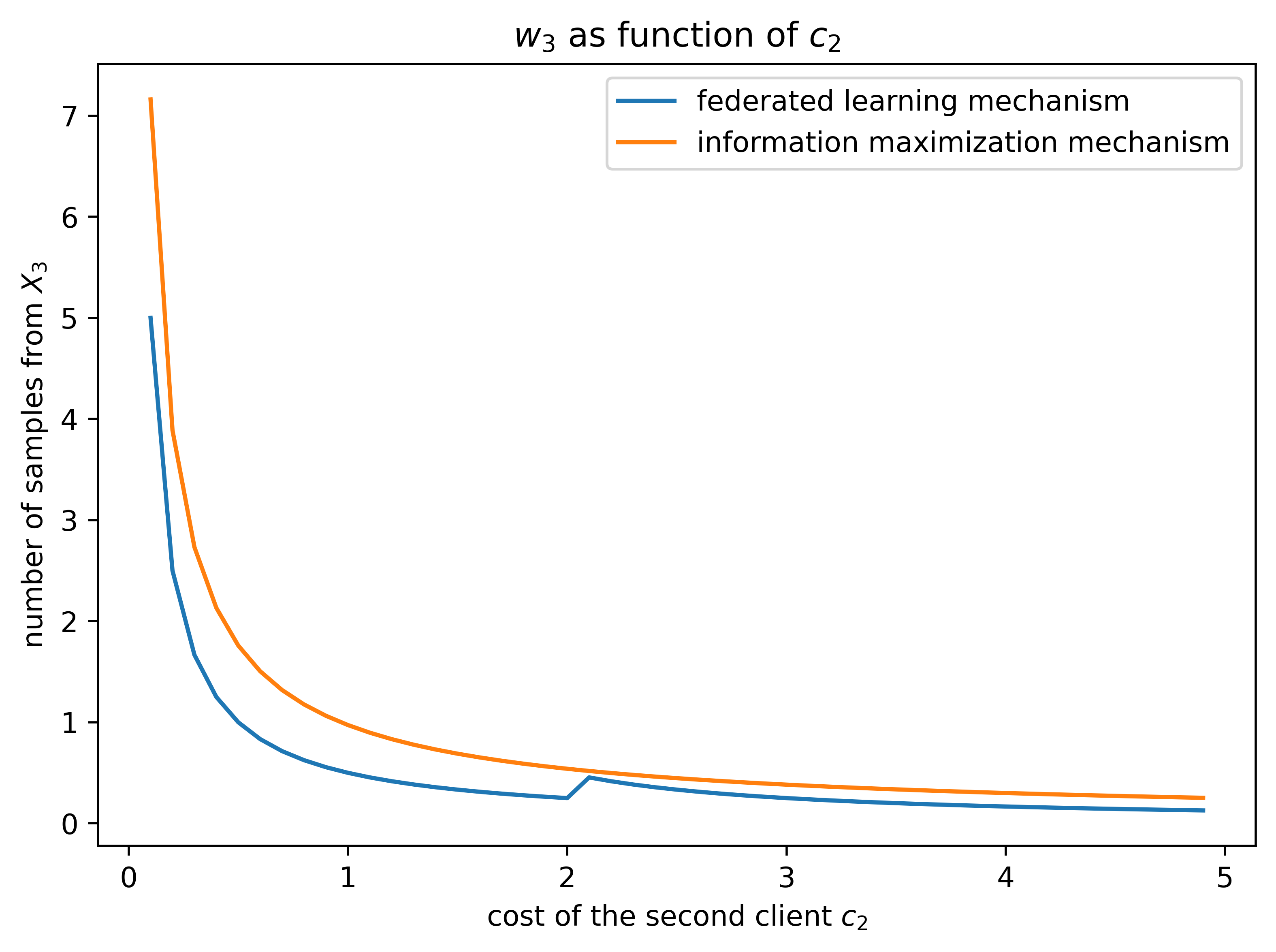}&
\includegraphics[width=0.4\textwidth]{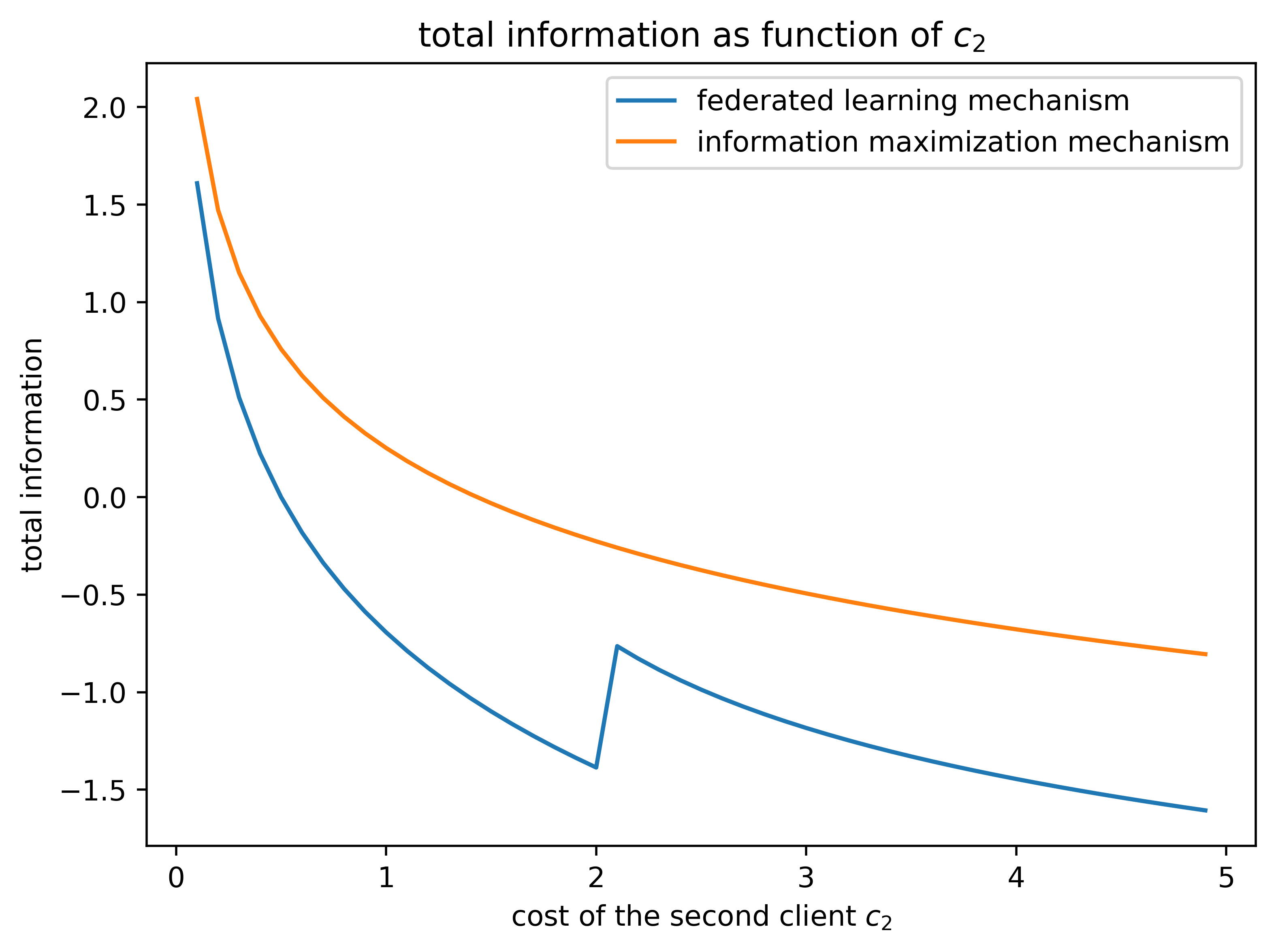}\\\vspace{-0.5em}
\scriptsize{$(c)$} & \scriptsize{$(d)$}
\end{tabular}%\vspace{-0.5em}
\caption{Comparison between federated learning mechanism and information maximization mechanism for different $c^{(2)}$ with fixed $c^{(1)} = 2, \theta = \pi/4$.}
\end{figure}

\begin{figure}[H]\label{fig:vary_theta}
\centering
\begin{tabular}{ccc} 
\includegraphics[width=0.4\textwidth]{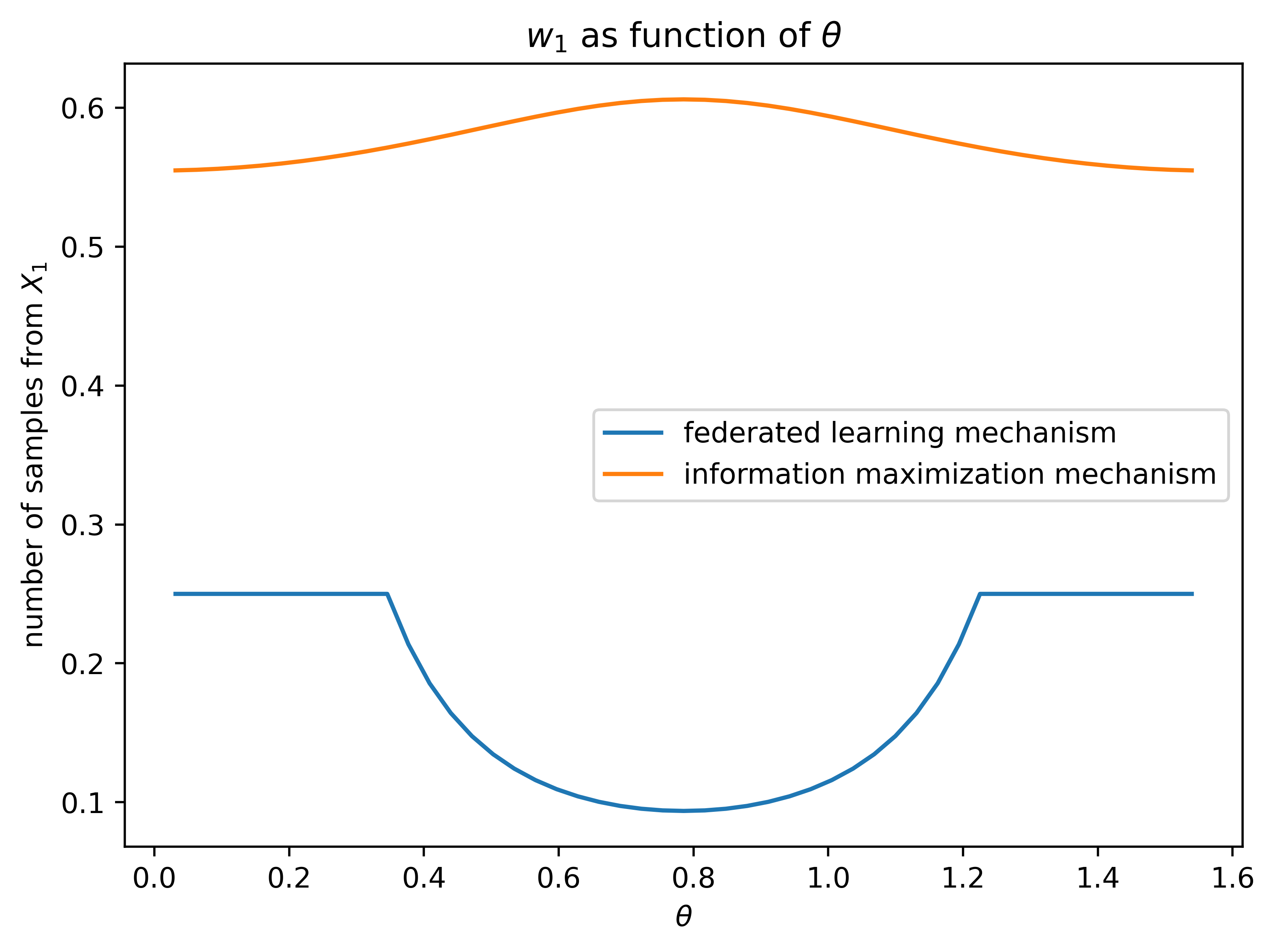} &
\includegraphics[width=0.4\textwidth]{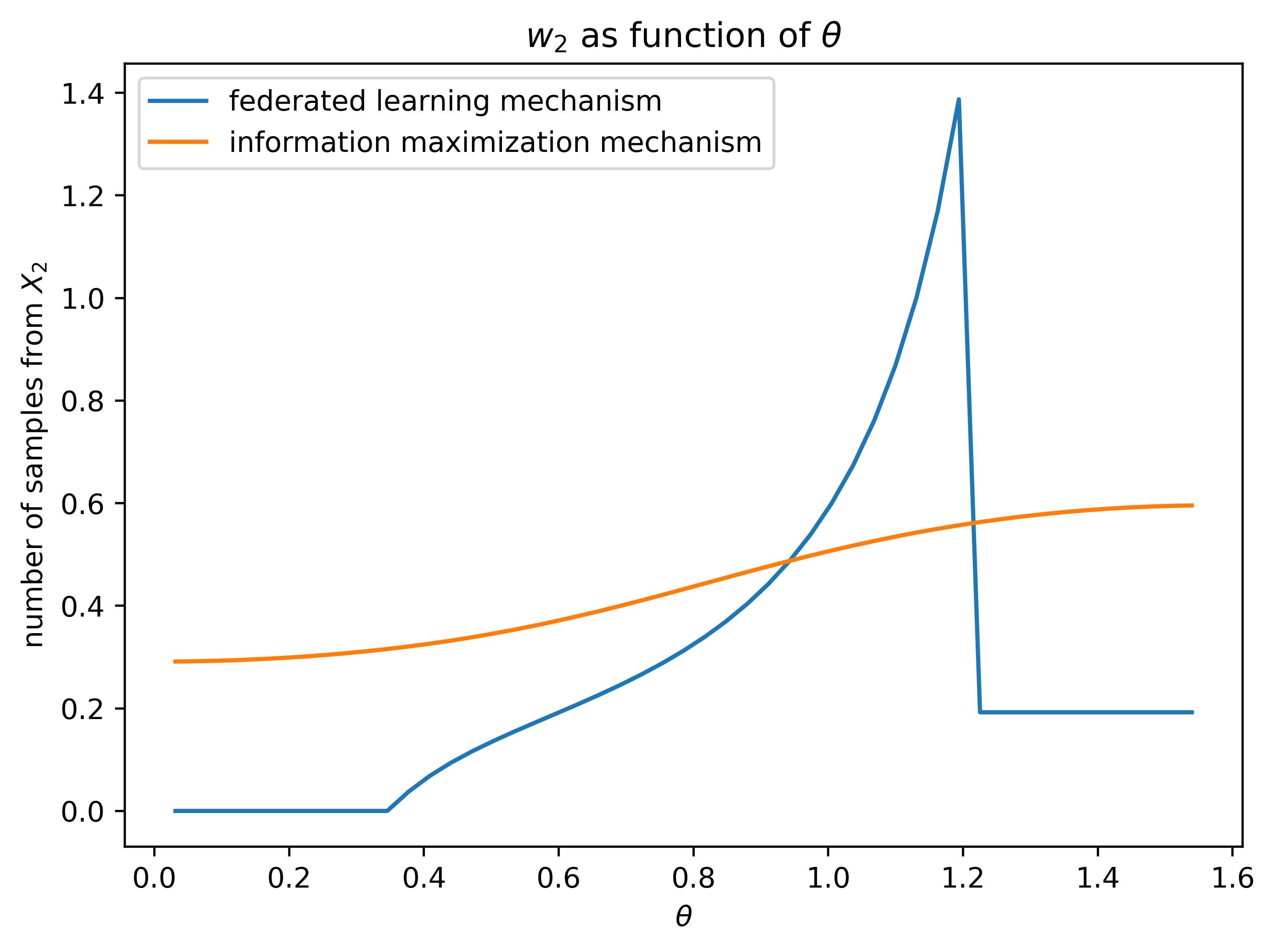}\\\vspace{-0.5em}
\scriptsize{$(a)$} & \scriptsize{$(b)$}\\
\\
\includegraphics[width=0.4\textwidth]{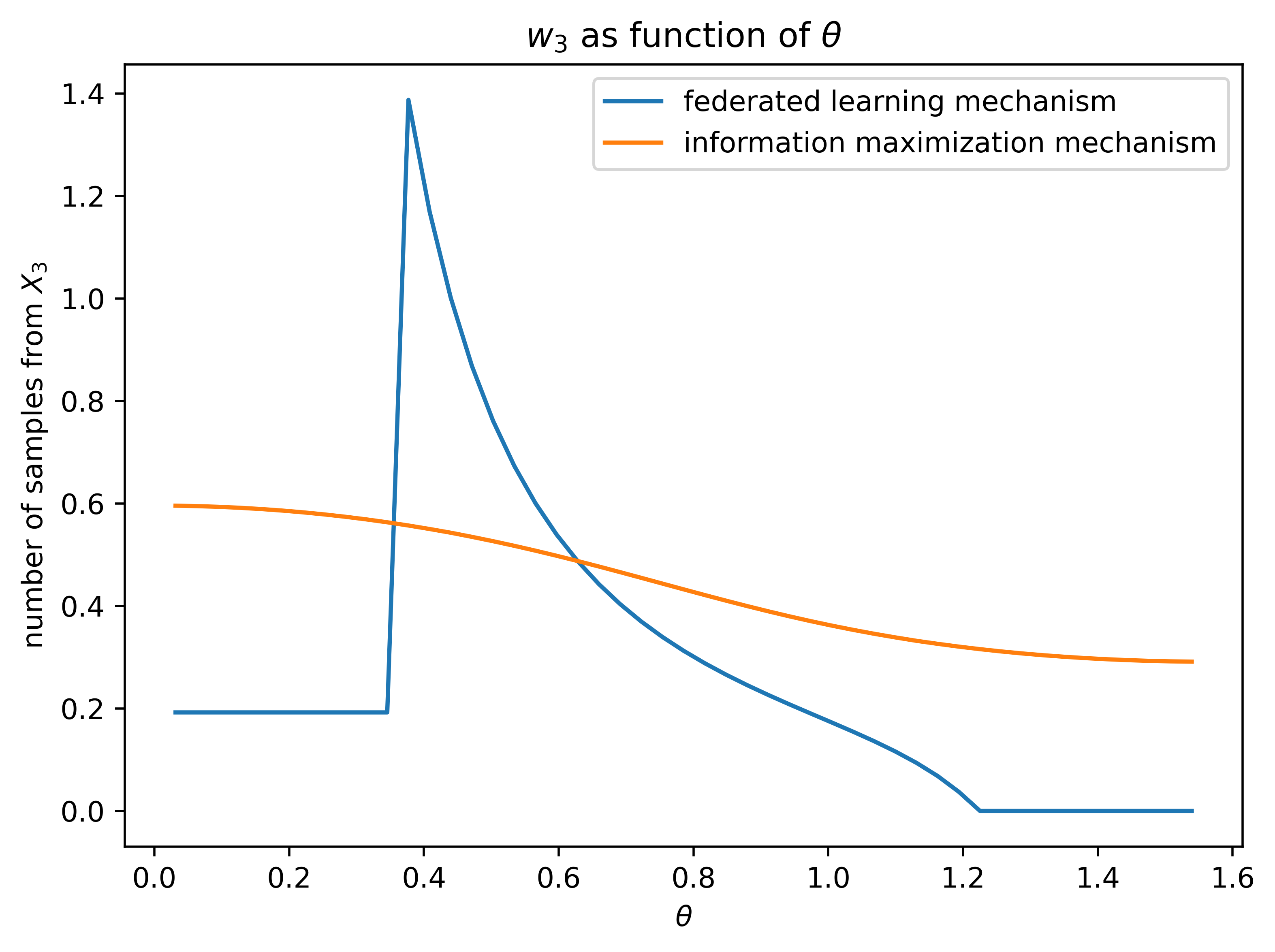}&
\includegraphics[width=0.4\textwidth]{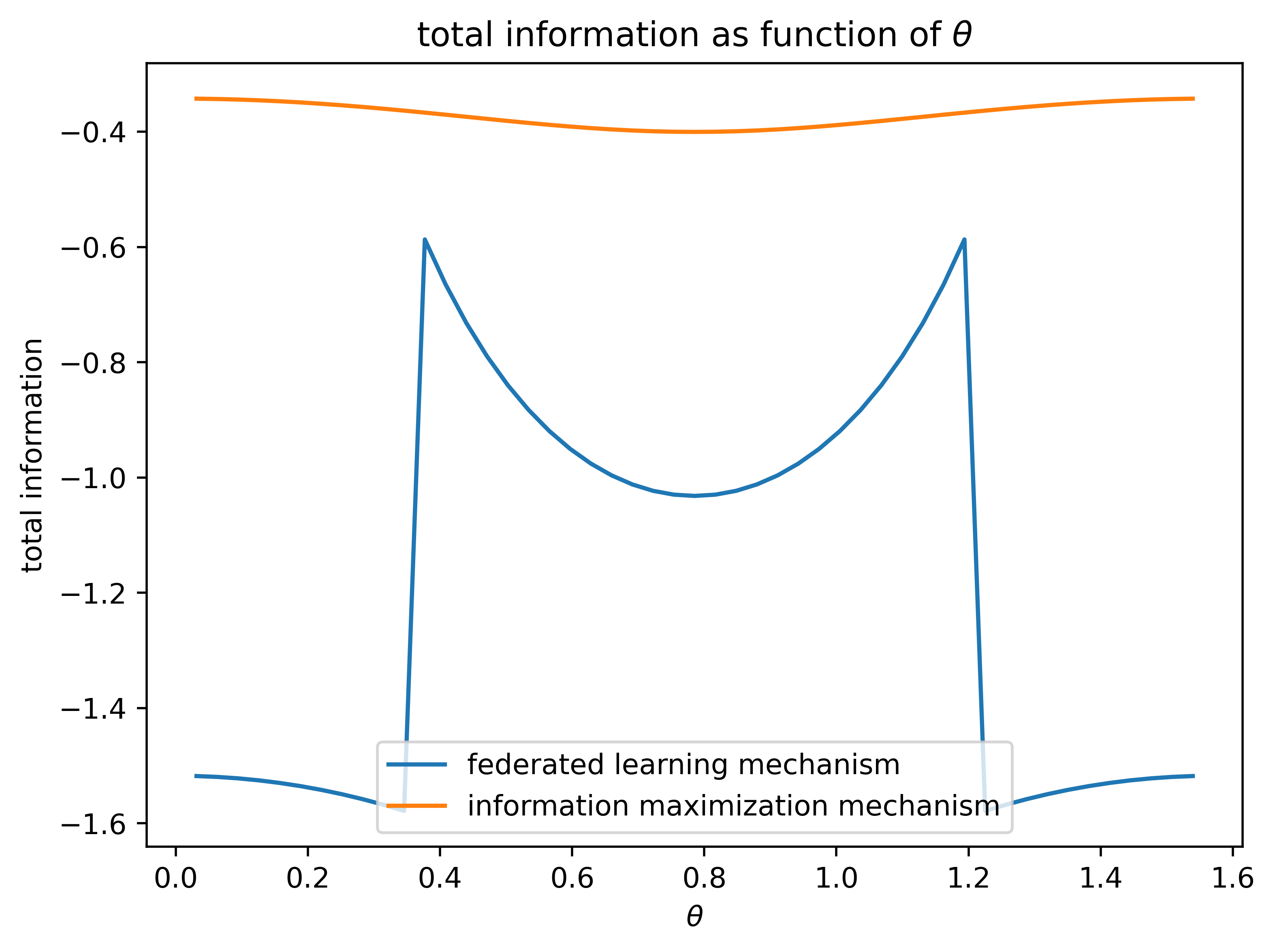}\\\vspace{-0.5em}
\scriptsize{$(c)$} & \scriptsize{$(d)$}
\end{tabular}%\vspace{-0.5em}
\caption{Comparison between federated learning mechanism and information maximization mechanism for different $\theta$ with fixed $c^{(1)} = 2, c^{(2)} = 2.5$.}
\end{figure}
\end{example}

\section{Omitted Proofs}

In this and the following sections, we will use $\supp(\cdot)$ to denote support of a distribution or vector, i.e., the set consisting of all indices corresponding to nonzero entries. Let $\mathbb{S}^d_+$ denote the set of symmetric positive semidefinite matrices in $\R^{d \times d}$, matrix Loewner order $\preceq$ is a partial ordering on $\mathbb{S}^{d}_+$, such that $A \preceq B$ iff $B-A \in \mathbb{S}^{d}_+$. Furthermore, $A \prec B$ if $B - A$ is positive definite. We overload this notation and say $x \preceq y$ for two vectors $x, y \in \R^d$ iff $x_i \leq y_i$ for all $i = 1,\dots,d$. Let $e_i$ denote the one-hot vector (whose dimension will be specified in the context) with the $i$-th coordinate being $1$ and the rest coordinate being zero. We also define $\mathbbm{1}(A) = \begin{cases}
    1, &~ A\\0, &~ \neg A
\end{cases}$.

\subsection{Proof of Proposition~\ref{cla:connection_to_optimal_design}}
\begin{proof}
Notice that when $w_{G_k^c} = 0$,
\begin{align}\label{eq:M_k-M}
    \M^{(k)}(w) = \left({A^{(k)}}^\top\left(\sum_{i \in G_k} w_ix_ix_i^\top\right)^{\dagger}A^{(k)}\right)^{-1} = \sum_{i \in G_k} w_i {A^{(k)}}^\top x_i ({A^{(k)}}^\top x_i)^\top.
\end{align}
(Notice that let $A^{(k)} = U \Lambda V^\top$ denote the Singular Value Decomposition (SVD) of $A^{(k)}$, $\sum_{i \in G_k} w_ix_ix_i^\top$ can be written as $U \begin{pmatrix}
    D & 0\\0 &0
\end{pmatrix}$ where $D \in \mathbb{S}^{r_k}_+$. Then
\begin{align*}
    \left({A^{(k)}}^\top\left(\sum_{i \in G_k} w_ix_ix_i^\top\right)^{\dagger}A^{(k)}\right)^{-1} = &~ \left(V \Lambda^\top U^\top U \begin{pmatrix}     D^{-1} &\mathbf{0}_{r_k \times (d-r_k)}\\\mathbf{0}_{(d-r_k) \times r_k}&\mathbf{0}_{(d-r_k) \times (d-r_k)} \end{pmatrix} U^\top U \Lambda V^\top\right)^{-1}\\
    = &~ \left(V  D^{-1} V^\top\right)^{-1}\\
    = &~ V  D V^\top\\
    = &~ {A^{(k)}}^\top\left(\sum_{i \in G_k} w_ix_ix_i^\top\right)A^{(k)}.
\end{align*}
This confirms Eq.~\eqref{eq:M_k-M}.)

It follows that
\begin{align*}
    u^{(k)}(w) = &~ f^{(k)}\left({\M^{(k)}(w)}^{-1}\right) - c \sum_{i\in G_k} w_i\\
    = &~ f^{(k)}\left(\left(\sum_{i \in G_k} w_i {A^{(k)}}^\top x_i ({A^{(k)}}^\top x_i)^\top\right)^{\dagger}\right) - c \sum_{i\in G_k} w_i.
\end{align*}
Now by changing of variable $w_i = \lambda \cdot \pi_i ~(\forall i \in G_k)$ where $\pi \in \Delta(G_k)$, we notice that the problem in Eq.~\eqref{eq:single_agent_utility} is equivalent to 
\begin{align*}
    \max &~ \lambda^{p}f^{(k)}\left(\left(\sum_{i \in G_k} \pi_i {A^{(k)}}^\top x_i ({A^{(k)}}^\top x_i)^\top\right)^{\dagger}\right) +g(\lambda) - c \lambda\\
    \text{s.t.} &~ \lambda > 0, \pi \in \Delta(G_k)
\end{align*}
The optimal $\pi_{G_k}^*$ is therefore the optimal design measure over the design space $\X_k$. As a result, the optimal design $w_{G_k}^*$ that maximizes $u^{(k)}$ is multiple of the optimal design measure over the design space $\left\{{A^{(k)}}^\top x: x \in \X_k \right\}$.
\end{proof}

\subsection{Proof of Proposition~\ref{cla:existence_ne}}
\begin{proof}
Define $\mathcal{W} = \left\{w \in \R^n_+: \M^{(k)}(w)\in \mathbb{S}^{r_k}_+\right\} $. Due to $\underset{\lambda \to +\infty}{\lim \sup} f^{(k)}\left((\lambda M)^{-1} \right)/(c \lambda) \leq 0$ we may constrain $w$ to a compact (and convex) subset of $\mathcal{W}$. Consider the best response mapping $B: \mathcal{W} \to \mathcal{W}$ where $B(w)_{G_k}$ is the unique unique maximizer of $(u^{(k)} \circ \mec^{(k)})(\cdot,w_{G_k^c})$. It follows that $B$ is a continuous, convex-valued and closed set-valued function from $\mathcal{W}$ to $2^{\mathcal{W}}$. Applying Kakutani's fixed-point theorem (Claim~\ref{cla:fixed_point_thm}), there exists fixed point $\wt w = B(\wt w)$. This $\wt w$ is a pure Nash equilibrium of the game.
\end{proof}

\begin{proposition}\label{prop:existence_ne_2}
    Suppose for all $k \in [K]$ and p.s.d. matrices $M$, $f^{(k)}(M^{-1})$ is a continuous non-decreasing concave function of $M$, $\underset{\lambda \to +\infty}{\lim \sup} f^{(k)}\left((\lambda M)^{-1} \right)/(c \lambda) \leq 0$ for any $c > 0$, and there exists $p_k \in \R$ and function $g_k: \R \to \R$ such that $f^{(k)}(\lambda M) = \lambda^{p_k} \cdot f^{(k)}(M) + g_k(\lambda)$. If the mechanism $\mec^{(k)}$ takes the form of $\mec^{(k)}(w) = \gamma_k(w) \cdot w$ where $\gamma_k : \R^n \to \R$ is a continuous positive function such that $\gamma_k(w)^{p_k}$ is non-increasing convex and $g_k(\gamma_k(w))$ is concave. Then there exists pure Nash equilibrium of the game defined by utilities $(u^{(k)} \circ \mec^{(k)})_{k \in [K]}$ and actions $(w_{G_k})_{k \in [K]}$, if any of the condition holds for all $k \in [K]$: (1) $f^{(k)}$ is negative; (2) $p_k = 0$.
\end{proposition}
\begin{proof}
Due to $\underset{\lambda \to +\infty}{\lim \sup} f^{(k)}\left((\lambda M)^{-1} \right)/(c \lambda) \leq 0$ we may constrain $w$ to a compact (and convex) set in $\R^n_+$. In what follows, we show that $u^{(k)} \circ \mec^{(k)}$ is quasi-concave in $w_{G_k}$ for all $k\in [K]$. First, notice that
\begin{align*}
    (u^{(k)} \circ \mec^{(k)})(w) = \gamma_k(w)^{p_k} \cdot f^{(k)}\left(\big({\M^{(k)}(w)}\big)^{-1}\right) + g_k(\gamma_k(w)) - c^{(k)}\sum_{i\in G_k}w_i.
\end{align*}
Define $\mathcal{W} = \left\{w \in \R^n_+: \M^{(k)}(w)\in \mathbb{S}^{r_k}_+\right\} $. For any $\lambda \in (0,1)$ and any $w_1,w_2 \in \mathcal{W}$ such that $w_{1,i} = w_{2,i}$ for all $i \notin G_k$, we have
\begin{align*}
    f^{(k)}\left(\big({\M^{(k)}(\lambda w_1 + (1-\lambda) w_2)}\big)^{-1}\right) = &~ f^{(k)}\left(\big(\lambda\M^{(k)}( w_1) + (1-\lambda)\M^{(k)} (w_2)\big)^{-1}\right)\\
    \geq &~ \lambda f^{(k)}\left(\big({\M^{(k)}(w_1)}\big)^{-1}\right) + (1-\lambda) f^{(k)}\left(\big({\M^{(k)}(w_2)}\big)^{-1}\right)
\end{align*}
where the first step comes from 
\begin{align*}
    \frac{\partial \M^{(k)}(w)}{\partial w_i} = &~ ({A^{(k)}}^\top\M(w)^{-1}A^{(k)})^{-1} {A^{(k)}}^\top\M(w)^{-1}x_ix_i^\top \M(w)^{-1}A^{(k)}({A^{(k)}}^\top\M(w)^{-1}A^{(k)})^{-1}\\
    = &~ {A^{(k)}}^\top x_ix_i^\top A^{(k)}, ~\forall i \in G_k
\end{align*}
and the second step comes from the concavity of $f^{(k)}$. Now combining this and the fact that $\gamma_k(w)^{p_k}$ is non-increasing convex, $\gamma_k(w)^{p_k} \cdot f^{(k)}\left(\big({\M^{(k)}(w)}\big)^{-1}\right)$ is concave in $w_{G_k}$ over $\mathcal{W}$ if any of the condition holds for all $k \in [K]$: (1) $f^{(k)}$ is negative; (2) $p_k = 0$. Since $g_k(\gamma_k(w))$ is concave and $- c^{(k)}\sum_{i\in G_k}w_i$ is linear, we conclude that $(u^{(k)} \circ \mec^{(k)})$ is concave in $w_{G_k}$ over $\mathcal{W}$. As $f^{(k)} = -\infty$ for singular input matrices, it follows that $u^{(k)} \circ \mec^{(k)}$ is quasi-concave in $w_{G_k}$.

Finally, applying Theorem 1 in \cite{dasgupta1986existence} establishes the existence of pure Nash equilibrium.
\end{proof}

\subsection{Proof of Proposition~\ref{prop:fed_efficient_allocation}}
\label{sec:proof_incentive_compatibility}

\begin{proof}
We first show that when $c^{(1)} = \cdots c^{(K)} = c$, $\wt w = \frac{d}{c} \cdot \pi^\ast$ is a pure Nash equilibrium. Consider the following function of $w_{G_k}$
\begin{align*}
    \bar u_k(w_{G_k}) = - \log\det \left((A^{(k)})^\top \M\left((w_{G_k},(\wt w_{G_k^c})\right)^{-1} A^{(k)}\right) - c \cdot  \sum_{i \in G_k} w_i
\end{align*}
Indeed, applying Lemma~\ref{lem:log_det_gradient} and Theorem~\ref{thm:equivalence_kiefer_wolfowitz},
\begin{align*}
    d \bar u_k(\wt w_{G_k}, \Delta w_{G_k})
    = &~ \sum_{i \in G_k} \Delta w_i \langle x_ix_i^\top, \M(\wt w)^{-1} \rangle - c \cdot \sum_{i \in G_k} \Delta w_i\\
    \leq &~ \left(\frac{d}{\|\wt w\|_1} - c \right)\cdot \sum_{i \in G_k} \Delta w_i\\
    = &~ 0,  ~\forall \Delta w_{G_k}.
\end{align*}
By concavity of $\bar u_k(\wt w_{G_k})$, $\wt w_{G_k}$ is the maximizer of $\bar u_k$.

To show IR, define
\begin{align*}
    w_{G_k}^\ast = \arg \max_{w_{G_k}} f^{(k)}_D\left( \left(\sum_{i \in G_k} w_i \cdot (A^{(k)})^\top x_ix_i^\top A^{(k)}\right)^{-1}\right) - c \sum_{i \in G_k} w_i,
\end{align*}
we have
\begin{align*}
    \left(u^{(k)}\circ \mec_{\fed}^{(k)}\right)\left(\wt w\right) \geq &~ \bar u_k(w_{G_k}^\ast)\\
    = &~ - \log\det \left((A^{(k)})^\top \M\left((w_{G_k}^\ast,\wt w_{G_k^c})\right)^{-1} A^{(k)}\right)   - c \cdot \sum_{i \in G_k} w_i^\ast\\
    \geq &~ - \log \det \left((A^{(k)})^\top (\sum_{i \in G_k} w_i^\ast x_ix_i^\top ) A^{(k)} \right)^{-1} - c \cdot \sum_{i \in G_k} w_i^\ast\\
    = &~ v^{(k)}_\ast
\end{align*}
where the first inequality follows from $\wt w_{G_k} \in \arg \max \bar u_k$; the second inequality comes from Lemma~\ref{lem:incentive_compatibility}.

Next, we show the uniqueness. Suppose $\wt w $ is a Nash equilibrium of the tuple of utility functions $\left(\left(u^{(k)}\circ \mec_{\fed}\right)\right)_{k \in [K]}$. 

It follows from first-order optimality that for any $\Delta w_{G_k}$ such that $\supp (\Delta w_{G_k}) \subset \supp (\wt w_{G_k})$
\begin{align*}
    0 = &~ d \bar u_k(\wt w_{G_k}, \Delta w_{G_k})\\
    = &~ \sum_{i \in G_k} \Delta w_i \langle x_ix_i^\top, \M(\wt w)^{-1} \rangle - c^{(k)} \cdot \sum_{i \in G_k} \Delta w_i\\
    = &~ \sum_{i \in G_k}  \left( \langle x_ix_i^\top, \M(\wt w)^{-1} \rangle - c \right)\cdot \Delta w_i.
\end{align*}
Therefore $  \langle x_ix_i^\top, \M(\wt w)^{-1} \rangle = c$ holds for all $i \in \supp(\wt w)$. Notice that
\begin{align*}
    d = &~ \left\langle \sum_{i=1}^n \wt w_i x_ix_i^\top, \M(\wt w)^{-1} \right\rangle \\
    = &~ c \sum_{i=1}^n \wt w_i.
\end{align*}
We thus have $\sum_{i=1}^n \wt w_i = \frac{d}{c}$, and as a result, $\langle x_ix_i^\top, \M(\wt w/\|\wt w\|_1)^{-1} \rangle = d$ holds for all $i \in \supp(\wt w)$. Furthermore, for any $i \notin \supp(\wt w)$ first-order optimality implies $\langle x_ix_i^\top, \M(\wt w/\|\wt w\|_1)^{-1} \rangle \leq d$. 
Applying Theorem~\ref{thm:equivalence_kiefer_wolfowitz}, we know that $\wt w/\|\wt w\|_1$ is a D-optimal design. This confirms that any strategic response follows the D-optimal design measure.

Finally, we show that federated learning is not efficient in any other criteria.
% \begin{proof}

\paragraph{\textbf{V-criterion.}}
Consider principal-agent experiment design with the accuracy function given by 
\begin{align*}
    f^{(k)}(w) = - \E_{x \sim p^{(k)}} \left[x^\top \M(w)^{-1}x\right].
\end{align*}
where $p^{(k)}$ represents the distribution of client $k$'s data and is supported on $G_k$.

In this case, the utility function under federated learning mechanism is given by 
\begin{align*}
    \left(u^{(k)}\circ \mec_{\fed}^{(k)}\right)(w) = - \E_{x \sim p^{(k)}} \left[x^\top \M(w)^{-1}x\right] - c \sum_{i\in G_k}w_i.
\end{align*}
The Nash equilibrium $w^\ast \in \R^n_+$ thus gives the following system:
\begin{align}\label{eq:ne_a_criterion}
    \E_{x \sim p^{(k)}} \left[\langle xx^\top, \M(w^\ast)^{-1} x_l x_l \M(w^\ast)^{-1}\rangle\right] \begin{cases}= c, &~  l \in \supp(w^\ast)\\ \leq c, &~ l \notin \supp(w^\ast)\end{cases} , ~\forall l \in G_k, k \in [K].
\end{align}
Consider the efficient allocation over the population $p$ ($\supp(p) \subset \X$)
\begin{align*}
    \pi^\ast = \arg \min_{\pi \in \Delta([n])} \E_{x \sim p} \left[x^\top \M(\pi)^{-1}x\right].
\end{align*}
The optimal design measure requires the following system:
\begin{align}\label{eq:a_optimal_design}
    \E_{x \sim p} \left[\langle xx^\top, \M(\pi^\ast)^{-1} x_l x_l \M(\pi^\ast)^{-1}\rangle\right] \begin{cases}= \langle \E_{x \sim p} \left[xx^\top\right], \M(\pi^\ast)^{-1}\rangle, &~  l \in \supp(\pi^\ast)\\ \leq \langle \E_{x \sim p} \left[xx^\top\right], \M(\pi^\ast)^{-1}\rangle, &~ l \notin \supp(\pi^\ast)\end{cases} , ~\forall l \in [n].
\end{align}
Therefore, if the unique pure Nash equilibrium follows the optimal design measure, then $p,p^{(1)},\dots,p^{(K)}$ must satisfy the linear system given by Eq.~\eqref{eq:ne_a_criterion} and Eq.~\eqref{eq:a_optimal_design}. The solution of this linear system is generally a subspace of $\Delta([n])$ that has zero measure. 

As a result, there exists a design measure such that federated learning mechanism is not efficient.

\paragraph{\textbf{G-criterion and E-criterion.}}

Consider the following design space $\X = \{(1,0,0)^\top, (0,1,0)^\top, (0,0,1)^\top\}$ and let there be two agents with index sets $G_1 = \{1\},G_2 = \{2,3\}$. It is not to see that in this case,
\begin{align*}
    f^{(k)}_G(w) = f^{(k)}_E(w) = - \max_{i \in G_k} \left[x_i^\top \M(w)^{-1}x_i\right] = &~ - \|(A^{(k)})^\top \M(w)^{-1}A^{(k)}\|_2\\
    = &~ \begin{cases}
        -w_1^{-1}, &~ k = 1\\ \min\{-w_2^{-1},-w_3^{-1}\}, &~ k = 2.
    \end{cases}
\end{align*}
Therefore the unique pure Nash equilibrium is given by $w_1 = c^{-1/2}, w_2 = w_3 = (2c)^{-1/2}$. This is clearly not proportional to the optimal design measure which is uniform over $\X$.

\paragraph{\textbf{A-criterion.}}

Consider the following design space $\X = \{(1,1)^\top, (1,0)^\top, (0,1)^\top\}$ and let there be two agents with index sets $G_1 = \{1,2\},G_2 = \{3\}$. It is not to see that in this case,
\begin{align*}
    -\tr \left[(A^{(k)})^\top \M(w)^{-1}A^{(k)}\right] =  \begin{cases}
        -\frac{2w_1+w_2+w_3}{w_1w_3+w_2w_3+w_1w_3}, &~ k = 1\\ -\frac{w_1+w_2}{w_1w_3+w_2w_3+w_1w_3}, &~ k = 2.
    \end{cases}
\end{align*}
The pure Nash equilibrium $(w_1,w_2,w_3)$ follows the system:
\begin{align*}
    \frac{(w_1+w_2)^2}{w_1w_3+w_2w_3+w_1w_3} - c = &~ 0\\
    \frac{w_2^2+w_3^2}{w_1w_3+w_2w_3+w_1w_3} - c = &~ 0\\
    \frac{2w_1^2+w_3^2+2w_1w_3}{w_1w_3+w_2w_3+w_1w_3} - c = &~ 0.
\end{align*}
The optimal design measure $(\pi_1,\pi_2,\pi_3)$ follows the system:
\begin{align*}
    \frac{2\pi_1^2+\pi_2^2+2\pi_2\pi_1}{\pi_1\pi_3+\pi_2\pi_3+\pi_1\pi_3} - \lambda = &~ 0\\
    \frac{\pi_2^2+\pi_3^2}{\pi_1\pi_3+\pi_2\pi_3+\pi_1\pi_3} - \lambda = &~ 0\\
    \frac{2\pi_1^2+\pi_3^2+2\pi_1\pi_3}{\pi_1\pi_3+\pi_2\pi_3+\pi_1\pi_3} - \lambda = &~ 0
\end{align*}
where $\lambda$ is the Lagrangian multiplier.

If$(w_1,w_2,w_3)$ is proportional to $(\pi_1,\pi_2,\pi_3)$, then comparing these two systems yields $\pi_1 = 0$. This is a contradiction.

\end{proof}

\subsection{Proof of Proposition~\ref{prop:free_riding_diversity}}
\begin{proof}
Suppose there exists a pure Nash equilibrium $\wt w$ such that $\wt w_j \neq 0$ and $j \in G_l$. Then from the first order optimality, for any $\Delta w_{G_l}$ such that $\supp (\Delta w_{G_l}) \subset \supp (\wt w_{G_l})$
\begin{align*}
    0 = &~ d \bar u_l(\wt w_{G_l}, \Delta w_{G_l})\\
    = &~ \sum_{i \in G_l} \Delta w_i \langle x_ix_i^\top, \M(\wt w)^{-1} \rangle - c \cdot \sum_{i \in G_l} \Delta w_i\\
    = &~ \sum_{i \in G_l}  \left( \langle x_ix_i^\top, \M(\wt w)^{-1} \rangle - c \right)\cdot \Delta w_i.
\end{align*}
It follows that $\langle x_jx_j^\top, \M(\wt w)^{-1} \rangle = c$. Similarly, $\langle x_ix_i^\top, \M(\wt w)^{-1} \rangle \leq c$ for all $i \in G_k$. From the condition, there exists $\alpha_i > 0$ such that $x_jx_j^\top = \sum_{i \in G_k} \alpha_i x_ix_i^\top$ and $\sum_{i \in G_k} \alpha_i < 1$. It follows that
\begin{align*}
    \langle x_jx_j^\top, \M(\wt w)^{-1} \rangle = \sum_{i \in G_k} \alpha_i \langle  x_ix_i^\top , \M(\wt w)^{-1} \rangle
    < c.
\end{align*}
This is a contradiction.
\end{proof}

\subsection{Proof of Proposition~\ref{prop:free_riding_cost}}
\begin{proof}
Suppose there exists a pure Nash equilibrium $\wt w$ such that $\wt w_j \neq 0$ and $j \in G_l$. Then first order optimality yields for any $\Delta w_{G_l}$ such that $\supp (\Delta w_{G_l}) \subset \supp (\wt w_{G_l})$
\begin{align*}
    0 = \sum_{i \in G_l}  \left( \langle x_ix_i^\top, \M(\wt w)^{-1} \rangle - c^{(l)} \right)\cdot \Delta w_i.
\end{align*}
It follows that $\langle x_jx_j^\top, \M(\wt w)^{-1} \rangle = c^{(l)}$. Similarly, $\langle x_ix_i^\top, \M(\wt w)^{-1} \rangle \leq c^{(k)}$ for all $i \in G_k$. This yields $c^{(k)} \geq \langle x_ix_i^\top, \M(\wt w)^{-1} \rangle= c^{(l)}$ for $i \in \supp(\wt w_{G_l})$, which contradicts the condition.
\end{proof}

\subsection{Proof of Proposition~\ref{prop:max_possible_information}}

\begin{proof}
We first notice that the problem in Eq.~\eqref{eq:max_object_design_linear} has compact feasible set and concave objective. Therefore, it has unique maximizer $w_{\max}$. 
Consider any mechanism $\mec$. The maximum possible information that can be achieved under $\mec$ is given by
\begin{align*}
    \max_{w \in \R^n_+} \log \det \M(w),
    ~\text{s.t.} \left(u^{(k)}\circ \mec^{(k)}\right)(w) \geq v^{(k)}_\ast.
\end{align*}
Let $\wt w$ be the maximizer of the above program, then due to $\mec^{(k)}(w)_i \leq w_i ,~\forall i$ we have
\begin{align*}
    u^{(k)}(w) \geq &~ \left(u^{(k)}\circ \mec^{(k)}\right)(w) \\
    \geq &~ v^{(k)}_\ast.
\end{align*}
Therefore $\wt w$ is in the feasible set of the optimization problem in Eq.~\eqref{eq:max_object_design_linear}. It follows from definition of $w_{\max}$ that $\log \det \M(\wt w) \leq \log \det \M(w_{\max})$.
\end{proof}

\subsection{Proof of Proposition~\ref{prop:free_riding}}
\label{sec:free_riding_proof}

\begin{proof}
First, notice that
\begin{align*}
    u^{(k)}\left(w_{\max}\right) = v^{(k)}_\ast, ~\forall k \in [K].
\end{align*}
Indeed, if there exist $k \in [K]$ such that $u^{(k)}\left(w_{\max}\right) > v^{(k)}_\ast$, then by setting $w_{\max,i}' = \begin{cases}
    (1+\epsilon) \cdot w_{\max,i}, &\text{if}~ i \in G_k\\
    w_{\max,i}, &\text{if}~ i \notin G_k
\end{cases}$ of sufficiently small $\epsilon > 0$, the constraints in Eq.~\eqref{eq:max_object_design_linear} is still satisfied, but $\log \det \M(w_{\max}') > \log \det \M(w_{\max})$. This contradicts to the fact that $w_{\max}$ is the maximizer.

Suppose $w_{\max}$ is the Nash equilibrium of $\left(\left(u^{(k)}\circ \mec_{\fed}^{(k)}\right)\right)_{k \in [K]}$, we will show that $\sum_{i=1}^K r_k = d$.

Indeed, by defining
\begin{align*}
    \bar u_k(w_{G_k}) := &~ - \log\det \left((A^{(k)})^\top \M\left((w_{G_k},w_{\max,G_k^c})\right)^{-1} A^{(k)}\right)   - c^{(k)} \cdot \sum_{i \in G_k} w_i,
\end{align*}
it follows that $w_{\max,G_k}$ is the maximizer of $\bar u_k$. 
First-order optimality condition and Lemma~\ref{lem:log_det_gradient} yields that for any $k$ and $l \in \supp(w_{\max,G_k})$
\begin{align*}
    0 = d\bar u_k(w_{\max,G_k},e_l) = \langle x_lx_l^\top, \M(w_{\max})^{-1}\rangle - c^{(k)}.
\end{align*}
As a result,
\begin{align}\label{eq:weighted_sum_w}
    \sum_{k=1}^K c^{(k)}\cdot \sum_{i\in G_k} w_{\max,i} = &~ \sum_{k=1}^K \sum_{i\in G_k} w_{\max,i}  \langle x_ix_i^\top, \M(w_{\max})^{-1}\rangle \notag\\
    = &~   \left\langle \sum_{i=1}^n w_{\max,i} x_ix_i^\top, \M(w_{\max})^{-1} \right\rangle \notag\\
    = &~ d.
\end{align}
Define
\begin{align*}
    v_k(w_{G_k}) = &~ u^{(k)}\left(\0,\dots,\0,w_{G_k},\0,\dots,\0\right)\\% - c^{(k)}\cdot \|w_{G_k}\|_1\\
    = &~ - \log \det \left((A^{(k)})^\top (\sum_{i \in G_k} w_i x_ix_i^\top ) A^{(k)} \right)^{-1} - c^{(k)} \cdot \sum_{i \in G_k} w_i.
\end{align*}
Let $w_{G_k}^\ast \in \arg \max_{w_{G_k}} v_k(w_{G_k})$, 
We have,
\begin{align*}
    \left(u^{(k)}\circ \mec_{\fed}^{(k)}\right)\left(w_{\max}\right) \geq &~ \bar u_k(w_{G_k}^\ast)\\
    = &~ - \log\det \left((A^{(k)})^\top \M\left((w_{G_k}^\ast,w_{\max,G_k^c})\right)^{-1} A^{(k)}\right)   - c^{(k)} \cdot \sum_{i \in G_k} w_i^\ast\\
    \geq &~ - \log \det \left((A^{(k)})^\top (\sum_{i \in G_k} w_i^\ast x_ix_i^\top ) A^{(k)} \right)^{-1} - c^{(k)} \cdot \sum_{i \in G_k} w_i^\ast\\
    = &~ v^{(k)}_\ast\\
    = &~ \left(u^{(k)}\circ \mec_{\fed}^{(k)}\right)\left(w_{\max}\right)
\end{align*}
where the second inequality is due to Lemma~\ref{lem:incentive_compatibility}. Therefore, the above inequalities are all equalities, which implies $w_{G_k}^\ast \in \arg \max \bar u_k(w_{G_k})$ and 
\begin{align*}
    - \log\det \left((A^{(k)})^\top \M\left(w_{\max}\right)^{-1} A^{(k)}\right) = - \log \det \left((A^{(k)})^\top (\sum_{i \in G_k} w_{\max,i} x_ix_i^\top ) A^{(k)} \right)^{-1} .
\end{align*}
It follows that $w_{\max,G_k} \in \arg \max v_k(w_{G_k})$ and thus $\|w_{\max,G_k}\|_1 = \|w_{G_k}^\ast\|_1$. 

First-order optimality condition and Theorem~\ref{thm:equivalence_kiefer_wolfowitz} yields that for any $k$ and $l \in G_k$
\begin{align*}
    0 = d\bar v_k(w_{G_k}^\ast,e_l) = \frac{r_k}{\|w_{G_k}^\ast\|_1} - c^{(k)}.
\end{align*}
As a result, $\|w_{\max,G_k}\|_1 = \|w_{G_k}^\ast\|_1 = \frac{r_k}{c^{(k)}}$. 
Combining this and Eq.~\eqref{eq:weighted_sum_w}, we have
\begin{align*}
    d = \sum_{k=1}^K c^{(k)}\cdot \sum_{i\in G_k} w_{\max,i} = \sum_{k=1}^K c^{(k)}\cdot \frac{r_k}{c^{(k)}} =  \sum_{k=1}^K r_k.
\end{align*}
This establishes the first statement. 

If $\sum_{k=1}^K r_k > d$, then the above arguments imply that there exist $k \in [K]$ and $i \in \supp(w_{\max,G_k})$ such that $d \bar u_k(w_{\max,G_k},e_i)<0$. It follows that by letting $\wt w_{G_k} = w_{\max,G_k} - \epsilon e_i$ for sufficiently small $\epsilon > 0$, we have
\begin{align*}
    \left(u^{(k)}\circ \mec^{(k)}_{\fed}\right)\left((\wt w_{G_k}, w_{\max,G_k^c})\right) > \left(u^{(k)}\circ \mec^{(k)}_{\fed}\right)\left( w_{\max}\right).
\end{align*}
This completes the proof.
\end{proof}

\subsection{Proof of Proposition~\ref{prop:data_max}}
\label{sec:max_proof}

\begin{proof}
Fix $k \in [K]$. Define 
\begin{align*}
    \bar u_k(w_{G_k}) := &~ - \log\det \left((A^{(k)})^\top \M\left((w_{G_k},w_{\max,G_k^c})\right)^{-1} A^{(k)}\right)  \\
    &~ - c^{(k)} \cdot \sum_{i \in G_k} w_i - c^{(k)} \cdot \sum_{i \in G_k}\left(w_{\max,i} - w_i\right)_+.
\end{align*}
To see that $w_{\max}$ is a pure NE, it suffices to show that $w_{\max,G_k} = \arg \max \bar u_k(w_{G_k})$. Indeed, if $\wt w_{G_k} = \arg \max \bar u_k(w_{G_k})$ and $\wt w_{G_k} \neq w_{\max,G_k}$. Consider the following two cases.

\textbf{Case 1: There exists $i \in G_k$ such that $\wt w_i < w_{\max,i}$.} 

Let $\wt w_{j}' = \begin{cases}
    w_{\max,i}, &~ \text{if}~ j = i\\
    \wt w_{j}, &~ \text{otherwise}
\end{cases} $. Then
\begin{align*}
    \bar u_k(\wt w_{G_k}') := &~ - \log\det \left((A^{(k)})^\top \M\left((\wt w_{G_k}',w_{\max,G_k^c})\right)^{-1} A^{(k)}\right)  \\
    &~ - c^{(k)} \cdot \left(w_{\max,i} + \sum_{j \in G_k/\{i\}} \wt w_j\right) - c^{(k)} \cdot \sum_{j \in G_k/\{i\}}\left(w_{\max,j} - \wt w_j\right)_+\\
    > &~ - \log\det \left((A^{(k)})^\top \M\left((\wt w_{G_k},w_{\max,G_k^c})\right)^{-1} A^{(k)}\right)  \\
    &~ - c^{(k)} \cdot \sum_{j \in G_k} \wt w_j - c^{(k)} \cdot \sum_{j \in G_k}\left(w_{\max,j} - \wt w_j\right)_+\\
    = &~ \bar u_k(\wt w_{G_k}),
\end{align*}
where the first step is due to $w_{\max,i} - \wt w_i' = 0$ and the second step comes from Lemma~\ref{lem:log_det_gradient} and $w_{\max,i} = \wt w_i + \left(w_{\max,i} - \wt w_i\right)_+$. 
This contradicts with $\wt w_{G_k} = \arg \max \bar u_k(w_{G_k})$.

\textbf{Case 2: $\wt w_j \geq w_{\max, j}, ~\forall j \in G_k$ and there exists $i \in G_k$ such that $\wt w_i > w_{\max,i}$.}

Notice that in this case $\log \det \M\left((\wt w_{G_k},w_{\max, G_k^c})\right) > \log \det \M\left(w_{\max}\right)$. Therefore there exists $j \in [K]$ such that $\left(u^{(j)}\circ \mec^{(j)}_{\max}\right)\left((\wt w_{G_j},w_{\max, G_j^c})\right) < v^{(j)}_\ast$, and it is obvious that such $j$'s must include $k$. As a result,
\begin{align*}
    \bar u_k(\wt w_{G_k}) = \left(u^{(k)}\circ \mec^{(k)}_{\max}\right)\left((\wt w_{G_k},w_{\max, G_k^c})\right) < v^{(k)}_\ast \leq \left(u^{(k)}\circ \mec^{(k)}_{\max}\right)\left(w_{\max}\right) = \bar u_k(w_{\max}).
\end{align*}
This is a contradiction. Therefore, we have shown that $(w_{\max,G_k})_{k \in [K]}$ is a pure Nash equilibrium. Since $w_{\max}$ is the solution of Eq.~\eqref{eq:max_object_design_linear}, Individual Rationality is satisfied. As a result, $w_{\max}$ is a strategic response of mechanism $\mec_{\max}$.

Next, we display uniqueness. 
Suppose for the sake of contradiction that there exists a Nash equilibrium $(\wt w_{G_k})_{k \in [K]} \neq (w_{\max,G_k})_{k \in [K]}$. We follow the above line of arguments and consider the following two cases.

\textbf{Case 1: There exists $k \in [K]$ and $i \in G_k$ such that $\wt w_i < w_{\max,i}$.} 

Define $\wt w_{j}'$ as follows:
\begin{align*}
\wt w_{j}' = \begin{cases}
    w_{\max,i}, &~ \text{if}~ j = i\\
    \wt w_{j}, &~ \text{otherwise}.
\end{cases}
\end{align*}
Then
\begin{align*}
    \left(u^{(k)}\circ \mec^{(k)}_{\max}\right)\left((\wt w_{G_k}',\wt w_{G_k^c})\right) := &~ - \log\det \left((A^{(k)})^\top \M\left((\wt w_{G_k}',\wt w_{G_k^c})\right)^{-1} A^{(k)}\right)  \\
    &~ - c^{(k)} \cdot \left(w_{\max,i} + \sum_{j \in G_k/\{i\}} \wt w_j\right) - c^{(k)} \cdot \sum_{j \in G_k/\{i\}}\left(w_{\max,j} - \wt w_j\right)_+\\
    > &~ - \log\det \left((A^{(k)})^\top \M\left(\wt w\right)^{-1} A^{(k)}\right)  \\
    &~ - c^{(k)} \cdot \sum_{j \in G_k} \wt w_j - c^{(k)} \cdot \sum_{j \in G_k}\left(w_{\max,j} - \wt w_j\right)_+\\
    = &~ \left(u^{(k)}\circ \mec^{(k)}_{\max}\right)\left(\wt w\right).
\end{align*}
This contradicts $\wt w_{G_k} = \arg \max_{w_{G_k}} \left(u^{(k)}\circ \mec^{(k)}_{\max}\right)\left((w_{G_k},\wt w_{G_k^c})\right)$.

\textbf{Case 2: $\wt w_j \geq w_{\max, j}, ~\forall j \in [n]$ and there exists $k \in [K]$ and $i \in G_k$ such that $\wt w_i > w_{\max,i}$.}

Since $\log \det \M\left(\wt w\right) > \log \det \M\left(w_{\max}\right)$, there exists $j \in [K]$ such that $\left(u^{(j)}\circ \mec^{(j)}_{\max}\right)\left(\wt w\right) < v^{(j)}_\ast$. Obviously, there exists $i \in G_j$ such that $\wt w_i > w_{\max,i}$. As a result,
\begin{align*}
    \left(u^{(j)}\circ \mec^{(j)}_{\max}\right)(\wt w) < v^{(j)}_\ast \leq \left(u^{(j)}\circ \mec^{(j)}_{\max}\right)\left(w_{\max}\right) \leq \left(u^{(j)}\circ \mec^{(j)}_{\max}\right)\left((w_{\max,G_j},\wt w_{G_j^c})\right).
\end{align*}
This means 
\begin{align*}
    \wt w_{G_j} \notin \arg \max_{w_{G_j} \in \R^{|G_j|}_+} \left(u^{(j)}\circ \mec^{(j)}_{\max}\right)\left(( w_{G_j},\wt w_{G_j^c})\right),
\end{align*} 
which is a contradiction.
\end{proof}

\subsection{Proof of Corollary~\ref{cor:incentive_compatibility}}\label{sec:proof_ic}

\begin{proof}
Suppose there exist $k \in [K]$ such that $\left(u^{(k)}\circ \mec_{\max}^{(k)}\right)\left(w_{\max}\right) > v^{(k)}_\ast$, then the definition of $\mec_{\max}$ yields $u^{(k)}(w_{\max}) = \left(u^{(k)}\circ \mec_{\max}^{(k)}\right)\left(w_{\max}\right) > v^{(k)}_\ast$. By setting $w_{i}' = \begin{cases}
    (1+\epsilon) \cdot w_{\max,i}, &\text{if}~ i \in G_k\\
    w_{\max,i}, &\text{if}~ i \notin G_k
\end{cases}$ of sufficiently small $\epsilon > 0$, we have for any $l \in [K]$,
\begin{align*}
    u^{(l)}(w') \geq v^{(l)}_*.
\end{align*}
Thus the constraints in Eq.~\eqref{eq:max_object_design_linear} is still satisfied, but $\log \det \M(w_{\max}') > \log \det \M(w_{\max})$. This contradicts the fact that $w_{\max}$ is the optimizer in Eq.~\ref{eq:max_object_design_linear}.
    
\end{proof}

\subsection{Proof of Proposition~\ref{prop:max_fairness}}

\begin{proof}

Fix $k \neq k'$ and assume $\left(u^{(k)}\circ \mec_{\max}^{(k)}\right)\left(\bar w\right) \geq \left(u^{(k')}\circ \mec_{\max}^{(k')}\right)\left(\bar w\right)$. By exchangeability and Corollary~\ref{cor:incentive_compatibility} we have, modulo a constant term $2\log \|x\|_2$, that
\begin{align*}
     \log \left(\|\bar w\|_1\right)  - c^{(k)} \cdot \|\bar w_{G_k}\|_1
    = &~ \left(u^{(k)}\circ \mec_{\max}^{(k)}\right)(\bar w)\\
    = &~ v^{(k)}_*\\
    = &~  -\log {c^{(k)}}  - 1.
\end{align*}
Therefore $c^{(k)}
 \leq c^{(k')}$ and we have
\begin{align*}
    \|\bar w_{G_k}\|_1 = \frac{ \log \left(\|\bar w\|_1\right) + \log {c^{(k)}}  + 1}{c^{(k)}}.
\end{align*}
Now we define
\begin{align*}
    f(c) = \frac{ \log \left(\|\bar w\|_1\right) + \log {c}  + 1}{c}.
\end{align*}
Notice that $f'(c) = -\frac{\log \left(\|\bar w\|_1\right) + \log {c}}{c^2} < 0$ for any $c \geq \min_{l \in [K]} c^{(l)}$, thus 
\begin{align*}
    \|\bar w_{G_k}\|_1 = &~ \frac{ \log \left(\|\bar w\|_1\right) + \log {c^{(k)}}  + 1}{c^{(k)}}\\
    \geq &~ \frac{ \log \left(\|\bar w\|_1\right) + \log {c^{(k')}}  + 1}{c^{(k')}}\\
    = &~ \|\bar w_{G_k'}\|_1.
\end{align*}
This confirms that $\|\bar w_{G_k}\|_1 \geq \|\bar w_{G_k'} \|_1$.
\end{proof}

\subsection{Proof of Proposition~\ref{prop:poa_max}}

\begin{proof}
Define $k_0 = \arg \min_{k\in [K]}c^{(k)}$ and
\begin{align*}
    \theta^{(k)}_* := &~ \max_{\pi \in \Delta([n])} - \log \det \left((A^{(k)})^\top \left(\sum_{i =1}^n \pi_ix_ix_i^\top \right)^{-1} A^{(k)}\right)\\
    \theta^{(k)} := &~ \max_{\pi \in \Delta(G_k)} \log \det \left((A^{(k)})^\top \left(\sum_{i \in G_k} \pi_ix_ix_i^\top \right) A^{(k)}\right).
\end{align*}
We have
\begin{align*}
    &~ \sg(w)\\
    = &~ \sum_{k=1}^K \left(u^{(k)}\circ \mec_{\max}^{(k)}\right)(w)\\
    = &~ \sum_{k=1}^K \left(- \log\det \left((A^{(k)})^\top \M(w)^{-1} A^{(k)}\right)  - c^{(k)} \cdot \sum_{i \in G_k} w_i - c^{(k)}\cdot\sum_{i \in G_k}\left(w_{\max,i} - w_i\right)_+\right)\\
    \leq &~ \sum_{k=1}^K \left( \theta_\ast^{(k)} + r_k\log \|w\|_1  - c^{(k)} \cdot \sum_{i \in G_k} w_i - c^{(k)}\cdot\sum_{i \in G_k}\left(w_{\max,i} - w_i\right)_+\right)\\
    \leq &~ \sum_{k=1}^K \left( \theta_\ast^{(k)} + r_k\log \frac{\sum_{k=1}^K r_k}{c^{(k_0)}}  - r_k - (c^{(k)} - c^{(k_0)}) \cdot \| w_{\max, G_k} \|_1 \right),
\end{align*}
where the maximizer in the last inequality is given by 
\begin{align*}
    \|w_{G_k}\|_1 = \begin{cases}
    \|w_{\max,G_k}\|_1 , &~ k \neq k_0\\
    \|w_{\max,G_k}\|_1 + \frac{\sum_{k=1}^K r_k}{c^{(k_0)}} - \| w_{\max} \|_1, &~ k = k_0.
\end{cases}
\end{align*}
Further, notice that
\begin{align*}
    \sg(w_{\max})
    = &~ \sum_{k=1}^K \left(u^{(k)}\circ \mec_{\max}^{(k)}\right)(w_{\max})\\
    = &~ \sum_{k=1}^K v^{(k)}_*\\
    = &~ \sum_{k=1}^K \left( \theta^{(k)} + r_k\log \frac{r_k}{c^{(k)}}  - r_k\right),
\end{align*}
where the second inequality uses Corollary~\ref{cor:incentive_compatibility}.

It follows that
\begin{align*}
    &~ \frac{\sg(w)}{\sg (w_{\max})}\\
    \leq &~ \frac{\sum_{k=1}^K  \left( \theta_\ast^{(k)} - \theta^{(k)}\right)}{\sum_{k=1}^K \left( \theta^{(k)} + r_k\log \frac{r_k}{c^{(k)}}  - r_k\right)} + \frac{\sum_{k=1}^K \left(r_k\log \frac{c^{(k)}\sum_{k=1}^K r_k}{r_kc^{(k_0)}} - (c^{(k)} - c^{(k_0)}) \cdot \| w_{\max, G_k} \|_1\right)}{\sum_{k=1}^K \left( \theta^{(k)} + r_k\log \frac{r_k}{c^{(k)}}  - r_k\right)}+1\\
    = &~ \frac{\sum_{k=1}^K  \Delta^{(k)}}{\sum_{k=1}^K \left( \theta^{(k)} + r_k\log \frac{r_k}{c^{(k)}}  - r_k\right)} + \frac{\sum_{k=1}^K \left(r_k\log \frac{c^{(k)}\sum_{k=1}^K r_k}{r_kc^{(k_0)}} - (c^{(k)} - c^{(k_0)}) \cdot \| w_{\max, G_k} \|_1\right)}{\sum_{k=1}^K \left( \theta^{(k)} + r_k\log \frac{r_k}{c^{(k)}}  - r_k\right)}+1.
\end{align*}
\end{proof}

\section{Efficiency under Heterogeneous Costs}\label{sec:efficient_allocation}

In Section~\ref{sec:fed_efficient}, we investigated the efficiency of federated learning with homogeneous costs. However, the proof of Proposition~\ref{prop:fed_efficient_allocation} demonstrates that federated learning is not efficient when costs are heterogeneous. Therefore, in this section, we focus on mechanism designs to incentivize efficient allocation. Specifically, we consider the objective $w_{\eff} = n_{\max} \cdot \pi^\ast$, where $\pi^\ast$ represents an optimal design measure under the D-criterion, and $n_{\max}$ is defined as 
\begin{align}\label{eq:efficient_allocation_object}
    &~ n_{\max} = \max_{n \in \R_+} ~n,\\
    \text{s.t.} &~ u^{(k)}(n \cdot \pi^\ast) \geq v^{(k)}_\ast, \forall k \in [K] \notag.
\end{align}
The objective $w_{\eff}$ aims to maximize the total number of data while preserving efficient allocation of experiments. However, the feasibility of the program defined by Eq.~\eqref{eq:efficient_allocation_object} is not guaranteed in general. We provide a result that establishes a condition under which $w_{\eff}$ is well-defined and lower bounded by $w^\ast_{G_k}$, where
\begin{align*}
    w^\ast_{G_k} =\arg \max - \log \det \left((A^{(k)})^\top (\sum_{i \in G_k} w_i x_ix_i^\top ) A^{(k)} \right)^{-1} - c^{(k)} \cdot \sum_{i \in G_k} w_i.
\end{align*}

\begin{assumption}[Data compatibility]\label{asp:data_compatibility}

We assume for any $k, k' \in [K]$, 
\begin{align*}
    u^{(k')}\left(\frac{\|w^\ast_{G_k}\|_1}{\|\pi^\ast_{G_k}\|_1} \cdot \pi^\ast\right) \geq v^{(k')}_\ast.
\end{align*}
\end{assumption}
This assumption implies that if we scale up the D-optimal design according to $w^\ast_{G_k}$, the utility for any other agent $k'$ is still no less than the maximum utility that agent $k'$ can achieve if she opts out of the collaborative learning and trains a model using her own data. Therefore, $\pi^\ast$ is compatible in the sense that no agent has an incentive to leave the collaborative learning program if they follow $\pi^\ast$ and each agent $k$ contribute at least $\|w^\ast_{G_k}\|_1$ data points. Under this condition, we can derive the following result:

\begin{proposition}[Feasibility and incentivized more contribution]\label{prop:feasibility}
Suppose Assumption~\ref{asp:data_compatibility} holds. Then the problem in Eq.~\eqref{eq:efficient_allocation_object} is feasible. 
Furthermore, For all $k \in [K]$ we have $n_{\max} \cdot \sum_{i \in G_k}\pi_i^\ast \geq \sum_{i \in G_k} w_i^\ast$.
\end{proposition}

\begin{proof}

Define $I_k = \left\{n \in \R_+: u^{(k)}(n \cdot \pi^\ast) \geq v^{(k)}_\ast\right\}$. Notice $u^{(k)}$ is concave and 
\begin{align*}
    &~ u^{(k)}\left(\frac{\|w^\ast_{G_k}\|_1}{\|\pi^\ast_{G_k}\|_1} \cdot \pi^\ast\right)\\
    = &~ - \log\det \left((A^{(k)})^\top \M\left(\frac{\|w^\ast_{G_k}\|_1}{\|\pi^\ast_{G_k}\|_1} \cdot \pi^\ast\right)^{-1} A^{(k)}\right)   - c^{(k)} \cdot  \|w^\ast_{G_k}\|_1\\
    \geq &~ - \log\det \left((A^{(k)})^\top \M\left(\left(w^\ast_{G_k},\left(\frac{\|w^\ast_{G_k}\|_1}{\|\pi^\ast_{G_k}\|_1} \cdot \pi^\ast_{j}\right)_{j \notin G_k}\right)\right)^{-1} A^{(k)}\right)   - c^{(k)} \cdot  \|w^\ast_{G_k}\|_1\\
    \geq &~ - \log\det \left((A^{(k)})^\top \left(\sum_{i \in G_k} w_i^\ast x_ix_i^\top \right) A^{(k)} \right)^{-1}   - c^{(k)} \cdot  \|w^\ast_{G_k}\|_1\\
    = &~ v_\ast^{(k)},
\end{align*}
where the second step comes from Applying Lemma~\ref{lem:log_det_gradient} and the fact that $u_k$ is concave wrt $w_{G_k}$; the third step comes from Lemma~\ref{lem:incentive_compatibility}. 
As a result, $I_k$ is a closed interval and $\frac{\|w^\ast_{G_k}\|_1}{\|\pi^\ast_{G_k}\|_1} \in I_k$ for any $k \in [K]$. We rewrite $I_k = [a_k,b_k]$ where $a_k \leq \frac{\|w^\ast_{G_k}\|_1}{\|\pi^\ast_{G_k}\|_1} \leq b$.

Assumption~\ref{asp:data_compatibility} implies that $\frac{\|w^\ast_{G_k}\|_1}{\|\pi^\ast_{G_k}\|_1} \in I_{k'}$ for any $k' \in [K]$. Therefore, $\cap_{k \in [k]} I_k \neq \emptyset$ and $n_{\max} = \min_{k \in [k]} b_k$. This establishes feasibility and $n_{\max} \geq \frac{\|w^\ast_{G_k}\|_1}{\|\pi^\ast_{G_k}\|_1}$.

\end{proof}

\subsection{Mechanism design for pure efficient allocation}
We begin by considering pure efficient allocation, which best illustrates the nature of the problem. In this subsection, we omit the cost functions and assume $c^{(1)} = \cdots = c^{(K)} = 0$. The goal in this section is to design mechanisms $\mec^{(k)}$ such that all Nash equilibrium wrt the utility functions $\left(\left(u^{(k)}\circ \mec^{(k)}\right)\right)_{k \in [K]}$ takes the form of $(\lambda \cdot \pi^\ast_{G_k})_{k \in [K]}$ where $\lambda > 0$, i.e. proportional to the optimal design measure.

We define the following mechanism based on scaling the design by a constant $\eta_k \leq 1$:
\begin{align}\label{eq:mechanism_pure_eff}
    \mec_{\peff}^{(k)}(w) = \eta_k {w} \text{ where } \eta_k^{-1} = \exp\left(\frac{d}{r_k} \cdot \left(\frac{(\sum_{i \notin G_k}\pi_i^\ast)(\sum_{i \in G_k}w_i)}{\sum_{i \notin G_k}w_i} -\sum_{i \in G_k}\pi_i^\ast \right)_+ \right).
\end{align}

The intuition behind $\eta_k$ is to introduce competition among agents, penalizing those who contribute proportionally less data than others. In fact, any strategic agent $k$ under this mechanism is incentivized to contribute no less than $\frac{\sum_{i \in G_k}\pi_i^\ast}{\sum_{i \notin G_k}\pi_i^\ast}$ times the total amount of data collected by the other agents. 
Therefore, the mechanism in Eq.~\eqref{eq:mechanism_pure_eff} ensures that the marginal probability of the aggregated design measure on each agent $k$, i.e., $(\sum_{i \in G_k}w_i)/(\sum_{i =1}^n w_i)$, aligns with the marginal probability of the optimal design measure, i.e. $(\sum_{i \in G_k}\pi^*_i)/(\sum_{i =1}^n \pi^*_i)$. By leveraging the properties of D-optimal design, we can demonstrate that it further ensures alignment between $w$ and $\pi^\ast$ for each coordinate. 
Besides subsampling, this mechanism can also be efficiently implemented by letting
\begin{align*}%\label{eq:contract}
    \hat \theta^{(k)} = \hat \theta + \zeta^{(k)},~\text{where}~
    \zeta^{(k)} \sim \mathcal{N}\left(0,(\eta_k^{-1}-1) \cdot \M(w)^{-1} \right).
\end{align*}

It follows that agent $k$'s utility is given by
\begin{align}\label{eq:utility_incentivized}
     \left(u^{(k)}\circ \mec^{(k)}_{\peff}\right)(w) = - \log\det \left((A^{(k)})^\top \M(w)^{-1} A^{(k)}\right)  - d \cdot \left(\frac{(\sum_{i \notin G_k}\pi_i^\ast)(\sum_{i \in G_k}w_i)}{\sum_{i \notin G_k}w_i} -\sum_{i \in G_k}\pi_i^\ast \right)_+ .
\end{align}

\begin{proposition}[Pure Efficient allocation]\label{prop:efficient_allocation}
For any $\lambda > 0$, $(\lambda \cdot \pi^\ast_{G_k})_{k \in [K]}$ is a pure Nash equilibrium of the tuple of utility functions $\left(\left(u^{(k)}\circ \mec^{(k)}_{\peff}\right)\right)_{k \in [K]}$. Furthermore, any pure Nash equilibrium takes the form of $(\lambda \cdot \pi^\ast_{G_k})_{k \in [K]}$.
\end{proposition}

\begin{proof}
Fix $k \in \mathbb{Z}_+,\lambda \in \R^+$ and $\bar w_i = \lambda \pi^\ast_i$ for all $i \notin G_k$. 

For the sake of brevity, we define the following function of $w_{G_k}$
\begin{align*}
    \bar u_k(w_{G_k}) = - \log\det \left((A^{(k)})^\top \M((w_{G_k},(\bar w_j)_{j \notin G_k}))^{-1} A^{(k)}\right) - {d} \cdot \left(\frac{(\sum_{i \notin G_k}\pi_i^\ast)(\sum_{i \in G_k}w_i)}{\sum_{i \notin G_k}\bar w_i} -\sum_{i \in G_k}\pi_i^\ast \right)_+
\end{align*}
where $(w_{G_k},(\bar w_j)_{j \notin G_k}) $ denotes the concatenation of $w_{G_k}$ and $(\bar w_j)_{j \notin G_k}$ such that 
\begin{align*}
    (w_{G_k},(\bar  w_j)_{j \notin G_k})_i = \begin{cases}
        w_i,&~ \text{if} ~i \in G_k\\
        \bar w_i,&~ \text{otherwise.}
    \end{cases}
\end{align*}
It suffices to show that $\bar w_{G_k}:= \lambda \pi^\ast_{G_k}$ is the unique maximizer of $\bar u_k(w_{G_k})$.

For any $\Delta w_{G_k}$, Lemma~\ref{lem:log_det_gradient} gives
\begin{align*}
    &~ d \bar u_k\left(\bar w_{G_k}, \Delta w_{G_k}\right) \\
    = &~  \left\langle \sum_{i \in G_k}\Delta w_i x_ix_i^\top , \M(\bar w)^{-1} \right\rangle - \frac{d(\sum_{i \notin G_k}\pi^\ast_i)}{\sum_{i \notin G_k} \bar w_i}\left(\sum_{i \in G_k}\Delta w_i\right)_+ \\
    = &~  \left\langle \sum_{i \in G_k}\Delta w_i x_ix_i^\top , \M(\pi^\ast)^{-1} \right\rangle \cdot \frac{\sum_{i \notin G_k} \pi^\ast_i}{\sum_{i \notin G_k} \bar w_i} - \frac{d(\sum_{i \notin G_k}\pi^\ast_i)}{\sum_{i \notin G_k} w_i}\left(\sum_{i \in G_k}\Delta w_i\right)_+ \\
    \leq &~ \sum_{i \in G_k}\Delta w_i  \left(\langle x_ix_i^\top , \M(\pi^\ast)^{-1} \rangle - d\right) \cdot \frac{\sum_{i \notin G_k} \pi^\ast_i}{\sum_{i \notin G_k} \bar w_i}\\
    \leq &~ 0,
\end{align*}
where the second step is due to $\M(\bar w)^{-1} = \M(\pi^\ast)^{-1}\cdot \frac{1}{\sum_{i=1}^n \bar w_i} = \M(\pi^\ast)^{-1}\cdot \frac{\sum_{i \notin G_k} \pi^\ast_i}{\sum_{i \notin G_k} \bar w_i}$; the last step uses $\langle x_ix_i^\top , \M(\pi^\ast)^{-1} \rangle \begin{cases} = d ,~& i \in \mathrm{supp}(\pi^\ast)\\ \leq d,~& i \notin \mathrm{supp}(\pi^\ast) \end{cases}$ by Theorem~\ref{thm:equivalence_kiefer_wolfowitz}. 
By concavity of $\bar u_k$, $\bar w_{G_k}:= \lambda \pi^\ast_{G_k}$ is the unique maximizer of $\bar u_k(w_{G_k})$. 

Therefore for any $\lambda > 0$, $(w_{G_k} = \lambda \cdot \pi^\ast_{G_k})_{k  \in [K]}$ is a Nash Equilibrium. 

In what follows, we show that any pure Nash Equilibrium takes the form of $(w_{G_k} = \lambda \cdot \pi^\ast_{G_k})_{k  \in [K]}, \lambda \in \R_+$. 

Suppose for the sake of contradiction a Nash equilibrium $(\wt w_{G_k})_{k \in [K]}$ not in the form of $(\lambda \cdot \pi^\ast_{G_k})_{k  \in [K]}$.

Fix $k \in [K]$. Consider the following utility as a function of $w_{G_k}$
\begin{align*}
    \bar u_k(w_{G_k}) = - \log\det \left((A^{(k)})^\top \M\left((w_{G_k},(\wt w_{G_k^c})\right)^{-1} A^{(k)}\right) - {d} \cdot \left(\frac{(\sum_{i \notin G_k}\pi_i^\ast)(\sum_{i \in G_k}w_i)}{\sum_{i \notin G_k}\wt w_i} -\sum_{i \in G_k}\pi_i^\ast \right)_+
\end{align*}
where $(w_{G_k},\wt w_{G_k^c} )$ denotes the concatenation of $w_{G_k}$ and $\wt w_{G_k^c}$ such that 
\begin{align*}
    (w_{G_k},\wt w_{G_k^c})_i = \begin{cases}
        w_i,&~ \text{if} ~i \in G_k\\
        \wt w_i,&~ \text{otherwise.}
    \end{cases}
\end{align*}

We assert that
\begin{align}\label{eq:marginal_prob_alignment}
    \frac{(\sum_{i \notin G_k}\pi_i^\ast)(\sum_{i \in G_k}\wt w_i)}{\sum_{i \notin G_k}\wt w_i} -\sum_{i \in G_k}\pi_i^\ast = 0,~\forall k \in [K].
\end{align}

Indeed, it is obvious that $\frac{(\sum_{i \notin G_k}\pi_i^\ast)(\sum_{i \in G_k}\wt w_i)}{\sum_{i \notin G_k}\wt w_i} -\sum_{i \in G_k}\pi_i^\ast \geq 0$ for all $k \in [K]$. (If there exists $k \in [K]$ such that $\frac{(\sum_{i \notin G_k}\pi_i^\ast)(\sum_{i \in G_k}\wt w_i)}{\sum_{i \notin G_k}\wt w_i} -\sum_{i \in G_k}\pi_i^\ast < 0$, then define $\hat w_{G_k} := (1+\epsilon)\wt w_{G_k}$. By applying Lemma~\ref{lem:log_det_gradient},
\begin{align*}
    \bar u_k(\hat w_{G_k}) = &~ - \log\det \left((A^{(k)})^\top \M\left(((1+\epsilon)\wt w_{G_k},(\wt w_{G_k^c})\right)^{-1} A^{(k)}\right)\\
    > &~ - \log\det \left((A^{(k)})^\top \M\left((\wt w_{G_k},(\wt w_{G_k^c})\right)^{-1} A^{(k)}\right) \\
    = &~ \bar u_k(\wt w_{G_k})
\end{align*}
holds for sufficiently small $\epsilon >0$. Contradiction!) Recall that $\pi^\ast \in \Delta[n]$. Examining $\frac{(\sum_{i \notin G_k}\pi_i^\ast)(\sum_{i \in G_k}\wt w_i)}{\sum_{i \notin G_k}\wt w_i} -\sum_{i \in G_k}\pi_i^\ast \geq 0$ for all $k \in [K]$ yields Eq.~\eqref{eq:marginal_prob_alignment}.

By Theorem~\ref{thm:equivalence_kiefer_wolfowitz}, there must exist $l \in [n]$ such that $\langle x_lx_l^\top,\M(\wt w/\|\wt w\|_1)^{-1} \rangle > d$. Suppose $l \in G_k$, then we have 
\begin{align*}
    d \bar u_k\left(\wt w_{G_k},e_l\right) 
    \geq &~ \left\langle x_lx_l^\top , \M(\wt w)^{-1} \right\rangle - \frac{d(\sum_{i \notin G_k}\pi^\ast_i)}{\sum_{i \notin G_k} \wt w_i}\\
    = &~ \left\langle x_lx_l^\top ,\M(\wt w/\|\wt w\|_1)^{-1} \right\rangle \cdot \frac{\sum_{i \notin G_k}\pi^\ast_i}{\sum_{i \notin G_k} \wt w_i} - \frac{d(\sum_{i \notin G_k}\pi^\ast_i)}{\sum_{i \notin G_k} \wt w_i}\\
    >&~ 0
\end{align*}
where the first step applies Lemma~\ref{lem:log_det_gradient}; the second step uses $\M(\wt w)^{-1} =\M(\wt w/\|\wt w\|_1)^{-1}\cdot \frac{1}{\sum_{i=1}^n \wt w_i} = \M(\wt w/\|\wt w\|_1)^{-1}\cdot \frac{\sum_{i \notin G_k} \pi^\ast_i}{\sum_{i \notin G_k} \wt w_i}$ due to Eq.~\eqref{eq:marginal_prob_alignment}. 
It follows that letting $\wt w' = \epsilon \cdot e_l + \wt w$ would increase $\bar u_k$ for sufficiently small $\epsilon > 0$. Contradiction!

\end{proof}

\subsection{Mechanism design for efficient allocation under different cost parameters}

We design the following feasible mechanism to achieve efficient allocation in Eq.~\eqref{eq:efficient_allocation_object} when the cost parameters are not the same. Define $\mec_{\eff}{(k)}$ as follows
\begin{align}\label{eq:mechanism_efficient_allocation}
    &~ \mec_{\eff}^{(k)}(w) = \rho_k {w},~\text{where} \\
    &~ \rho_k^{-1} = \exp\left(\frac{c^{(k)}}{r_k} \cdot \sum_{i \in G_k} \left( n_{\max} \cdot \pi_i^\ast - w_i\right)_+ + \sum_{i \in G_k} \left(\frac{(\sum_{i \notin G_k}\pi_i^\ast)w_i}{(\sum_{i \notin G_k}\bar w_i)\pi^*_i} - 1 \right)_+ \cdot \mathbbm{1}\left(w_i \geq n_{\max} \cdot \pi_i \right)\right) \notag
\end{align}
In $\rho_k$, the first term is the same as the regularization term in $\mec_{\max}$ and functions as incentivizing more data contribution; the second term is similar to the regularization term in $\mec_{\peff}$ and serves as incentivizing alignment with optimal design measure. Therefore, although the objective $w_{\eff}$ may take complex forms, these two simple terms together create the incentive for each agent to follow the optimal design measure while increasing the total information. 
Besides subsampling, this mechanism can also be efficiently implemented by letting
\begin{align*}%\label{eq:contract_cost}
    \hat \theta^{(k)} = \hat \theta + \zeta^{(k)}, ~
    \text{where}~ \zeta^{(k)} \sim\mathcal{N}\left(0,(\rho_k^{-1}-1) \cdot \M(w)^{-1} \right)\notag.
\end{align*}
The $k$-th agent's utility is then given by
\begin{align}\label{eq:utility_incentivized_cost}
    \left(u^{(k)}\circ \mec_{\eff}^{(k)}\right)(w) = &~ - \log\det \left((A^{(k)})^\top \M(w)^{-1} A^{(k)}\right)   - c^{(k)} \cdot \sum_{i \in G_k} \left( n_{\max} \cdot \pi_i^\ast - w_i\right)_+ \\
    &~ - r_k\cdot\sum_{i \in G_k} \left(\frac{(\sum_{i \notin G_k}\pi_i^\ast)w_i}{(\sum_{i \notin G_k} w_i)\pi^*_i} - 1 \right)_+ \cdot \mathbbm{1}\left(w_i \geq n_{\max} \cdot \pi_i \right)- c^{(k)} \left(\sum_{i \in G_k} w_i\right)\notag.
\end{align}

\begin{proposition}[Data maximization and efficient allocation]\label{prop:data_max_efficient_allocation}
The efficient allocation design $(w_{\eff,G_k})_{k \in [K]}$ is the unique strategic response to the mechanism $\mec_{\eff}$.

\end{proposition}

\begin{proof}
Fix $k$ and $\bar w_i = n_{\max} \pi^\ast_i$ for all $i \notin G_k$. 
Define $N = \sum_{i \notin G_k}\bar w_i x_ix_i^\top$. 
Define the following function of $w_{G_k}$
\begin{align*}
    &~ \bar u_k(w_{G_k})\\
    = &~ - \log\det \left((A^{(k)})^\top \M((w_{G_k},\bar w_{G_k^c}))^{-1} A^{(k)}\right) - c^{(k)} \cdot \left(\sum_{i \in G_k} w_i\right) - c^{(k)} \cdot \sum_{i \in G_k} \left( n_{\max} \cdot \pi_i^\ast - w_i\right)_+\notag \\
    &~ - {r_k} \cdot \sum_{i \in G_k} \left(\frac{(\sum_{i \notin G_k}\pi_i^\ast)w_i}{(\sum_{i \notin G_k}\bar w_i)\pi^*_i} - 1 \right)_+ \cdot \mathbbm{1}\left(w_i \geq n_{\max} \cdot \pi_i \right)
\end{align*}
where $(w_{G_k},\bar w_{G_k^c}) $ denotes the concatenation of $w_{G_k}$ and $\bar w_{G_k^c}$ such that 
\begin{align*}
    (w_{G_k},\bar w_{G_k^c})_i = \begin{cases}
        w_i,&~ \text{if} ~i \in G_k\\
        \bar w_i,&~ \text{otherwise.}
    \end{cases}
\end{align*}
To show that $w_{\eff}$ is a pure Nash equilibrium, it suffices to show that $\bar w_{G_k}:= n_{\max} \pi^\ast_{G_k}$ is the unique maximizer of $\bar u_k(w_{G_k})$.

Indeed, consider any $w_{G_k} \neq \bar w_{G_k}$. 

\textbf{Case 1: there exists $i \in G_k$ such that $w_i < n_{\max}\cdot \pi^\ast_i$. }

For all $j \in G_k$ let $\wt w_{j} = \begin{cases}
    w_j, ~&\text{if}~ w_j \geq n_{\max}\cdot \pi^\ast_j\\
    n_{\max} \cdot \pi_j^*, ~&\text{if}~ w_j < n_{\max}\cdot \pi^\ast_j
\end{cases}$.
From Lemma~\ref{lem:log_det_gradient},
\begin{align*}
    &~ \bar u_k(\wt w_{G_k})\\
    = &~ - \log\det \left((A^{(k)})^\top \M((\wt w_{G_k},\bar w_{G_k^c}))^{-1} A^{(k)}\right) - c^{(k)} \cdot \left(\sum_{i \in G_k} \wt w_i\right) - c^{(k)} \cdot \sum_{i \in G_k} \left( n_{\max} \cdot \pi_i^\ast - \wt w_i\right)_+\notag \\
    &~ - {r_k} \cdot \sum_{i \in G_k} \left(\frac{(\sum_{i \notin G_k}\pi_i^\ast)\wt w_i}{(\sum_{i \notin G_k}\bar w_i)\pi^*_i} - 1 \right)_+ \cdot \mathbbm{1}\left(\wt w_i \geq n_{\max} \cdot \pi_i \right)\\
    > &~ - \log\det \left((A^{(k)})^\top \M(( w_{G_k},\bar w_{G_k^c}))^{-1} A^{(k)}\right) - c^{(k)} \cdot \left(\sum_{i \in G_k} w_i\right) - c^{(k)} \cdot \sum_{i \in G_k} \left( n_{\max} \cdot \pi_i^\ast - w_i\right)_+\notag \\
    &~ - {r_k} \cdot \sum_{i \in G_k} \left(\frac{(\sum_{i \notin G_k}\pi_i^\ast)w_i}{(\sum_{i \notin G_k}\bar w_i)\pi^*_i} - 1 \right)_+ \cdot \mathbbm{1}\left(w_i \geq n_{\max} \cdot \pi_i \right)\\
    = &~ \bar u_k(w_{G_k}).
\end{align*}
This yields a contradiction.

\textbf{Case 2: $w_i \geq n_{\max}\cdot \pi^\ast_i$ for all $i \in G_k$. }

Define $\epsilon_j = \frac{w_{j}}{n_{\max}\cdot \pi^\ast_j} - 1$ for all $j \in G_k$, then $\min_{i\in G_k}\epsilon_i \geq 0$ and $\max_{i\in G_k}\epsilon_i > 0$. We have
\begin{align*}
    &~ \bar u_k(w_{G_k})\\
    = &~ - \log\det \left((A^{(k)})^\top \left(\sum_{i \in G_k} (1+\epsilon_i)\bar w_i x_ix_i^\top + N \right)^{-1} A^{(k)}\right)- c^{(k)} \cdot \sum_{i \in G_k} \left( n_{\max} \cdot \pi_i^\ast - w_i\right)_+\notag \\
    &~  - c^{(k)} \cdot \left(\sum_{i \in G_k} w_i\right) - {r_k} \cdot \sum_{i \in G_k} \epsilon_i\\
    \leq &~ - \log\det \left(\left(1+\max_{i\in G_k}\epsilon_i\right)^{-1}(A^{(k)})^\top \left(\sum_{i \in G_k} \bar w_i x_ix_i^\top + N \right)^{-1} A^{(k)}\right)\\
    &~ - c^{(k)} \cdot \sum_{i \in G_k} \left( n_{\max} \cdot \pi_i^\ast - \bar w_i\right)_+\notag  - c^{(k)} \cdot \left(\sum_{i \in G_k} \bar w_i\right) - {r_k} \cdot \sum_{i \in G_k} \epsilon_i\\
    = &~ \bar u_k(\bar w_{G_k}) + {r_k} \cdot \log\left(1+\max_{i\in G_k}\epsilon_i\right) - {r_k} \cdot \sum_{i \in G_k} \epsilon_i\\
    < &~ \bar u_k(\bar w_{G_k}).
\end{align*}
It follows that $\bar u_k(w_{G_k}) < \bar u_k(\bar w_{G_k})$, also a contradiction. 

Combining the above two cases confirms that $\left(n_{\max} \cdot \pi^\ast_{G_k} \right)_{k \in [K]}$ is a pure NE. By definition of $w_{\eff}$ in Eq.~\eqref{eq:efficient_allocation_object}, Individual rationality is satisfied. Therefore $w_{\eff}$ is a strategic response of $\mec_{\eff}$.

In what follows, we show that $n_{\max} \pi^\ast$ is the unique pure Nash equilibrium. 
Consider any pure Nash equilibrium $(\wt w_{G_k})_{k \in [K]}$, we will show in the following three steps that it must be equal to $(n_{\max} \cdot \pi^\ast_{G_k})_{k  \in [K]}$.

For any $k \in [K]$ define the following utility as a function of $w_{G_k}$
\begin{align*}
    &~ \bar u_k(w_{G_k})\\
    = &~ - \log\det \left((A^{(k)})^\top \M((w_{G_k},\wt w_{G_k^c}))^{-1} A^{(k)}\right) - c^{(k)} \cdot \left(\sum_{i \in G_k} w_i\right) - c^{(k)} \cdot \sum_{i \in G_k} \left( n_{\max} \cdot \pi_i^\ast - w_i\right)_+\notag \\
    &~ - {r_k} \cdot \sum_{i \in G_k} \left(\frac{(\sum_{i \notin G_k}\pi_i^\ast)w_i}{(\sum_{i \notin G_k}\wt w_i)\pi^*_i} - 1 \right)_+ \cdot \mathbbm{1}\left(w_i \geq n_{\max} \cdot \pi_i \right)
\end{align*}
where $(w_{G_k},\wt w_{G_k^c}) $ denotes the concatenation of $w_{G_k}$ and $\wt w_{G_k^c}$ such that 
\begin{align*}
    (w_{G_k},\wt w_{G_k^c})_i = \begin{cases}
        w_i,&~ \text{if} ~i \in G_k\\
        \wt w_i,&~ \text{otherwise.}
    \end{cases}
\end{align*}
It follows that $\bar u_k(\wt w_{G_k}) = \max_{w_{G_k}}\bar u_k(w_{G_k}),~\forall k \in [K]$.

\textbf{Step 1.} We first show that $\wt w_i \geq n_{\max} \pi^\ast_i$ for any $i \in [n]$. 

Indeed, if there exists $k \in [K]$ and $i \in G_k$ such that $\wt w_i < n_{\max} \pi^\ast_i$. Let $\hat w_{j} = \begin{cases}
\wt w_j, ~&\text{if}~ \wt w_j \geq n_{\max}\cdot \pi^\ast_j\\
    n_{\max} \cdot \pi_j^*, ~&\text{if}~ \wt w_j < n_{\max}\cdot \pi^\ast_j
\end{cases}$.
From Lemma~\ref{lem:log_det_gradient},
\begin{align*}
    &~ \bar u_k(\hat w_{G_k})\\
    = &~ - \log\det \left((A^{(k)})^\top \M((\hat w_{G_k},\wt w_{G_k^c}))^{-1} A^{(k)}\right) - c^{(k)} \cdot \left(\sum_{i \in G_k} \hat w_i\right) - c^{(k)} \cdot \sum_{i \in G_k} \left( n_{\max} \cdot \pi_i^\ast - \hat w_i\right)_+\notag \\
    &~ - {r_k} \cdot \sum_{i \in G_k} \left(\frac{(\sum_{i \notin G_k}\pi_i^\ast)\hat w_i}{(\sum_{i \notin G_k}\wt w_i)\pi^*_i} - 1 \right)_+ \cdot \mathbbm{1}\left(\hat w_i \geq n_{\max} \cdot \pi_i \right)\\
    > &~ \bar u_k(\wt w_{G_k}).
\end{align*}
This contradicts with the fact that $\wt w_{G_k}$ is the maximizer of $\bar u_k$.

\paragraph{Step 2.} We show that there exists $\lambda \geq n_{\max}$ such that $\wt w = \lambda \cdot \pi^\ast$.

Suppose $\exists ~ k \in [K]$ and $i \in G_k$ such that $\frac{(\sum_{j \notin G_k}\pi_i^\ast)\wt w_i}{\sum_{j \notin G_k}\wt w_i} > \pi_i^\ast $, define $\epsilon_i = \frac{(\sum_{j \notin G_k}\pi_j^\ast)\wt w_i}{(\sum_{j \notin G_k}\wt w_j)\pi^*_i} - 1 > 0$. Let $\hat w_{j} = \begin{cases}
\wt w_j, ~&\text{if}~ j \neq i \\
   \frac{(\sum_{j \notin G_k}\wt w_j)\pi^*_i}{\sum_{j \notin G_k}\pi_j^\ast}, ~&\text{if}~ j = i
\end{cases}$. We have
\begin{align*}
    &~ \bar u_k(\wt w_{G_k})\\
    = &~ - \log\det \left((A^{(k)})^\top \left((1+\epsilon_i)\hat w_i x_ix_i^\top + \sum_{j \neq i } \wt w_j x_jx_j^\top  \right)^{-1} A^{(k)}\right)- c^{(k)} \cdot \left(\sum_{j \in G_k} \wt w_j\right)\\
    &~ - c^{(k)} \cdot \sum_{j \in G_k} \left( n_{\max} \cdot \pi_j^\ast - \wt w_j\right)_+  - {r_k} \cdot \sum_{j \in G_k} \left(\frac{(\sum_{l \notin G_k}\pi_l^\ast)\wt w_j}{(\sum_{l \notin G_k}\wt w_l)\pi^*_j} - 1 \right)_+ \cdot \mathbbm{1}\left(\hat w_j \geq n_{\max} \cdot \pi_j \right)\\
    \leq &~ - \log\det \left(\left(1+\epsilon_i\right)^{-1}(A^{(k)})^\top \left(\sum_{j \in G_k} \hat w_j x_jx_j^\top \right)^{-1} A^{(k)}\right)- {r_k} \cdot \epsilon_i- c^{(k)} \cdot \left(\sum_{j \in G_k} \hat w_j\right)\\
    &~ - c^{(k)} \cdot \sum_{j \in G_k} \left( n_{\max} \cdot \pi_j^\ast - \hat w_j\right)_+ - {r_k} \cdot \sum_{j \in G_k, j \neq i} \left(\frac{(\sum_{l \notin G_k}\pi_l^\ast)\hat w_j}{(\sum_{l \notin G_k}\wt w_l)\pi^*_j} - 1 \right)_+ \cdot \mathbbm{1}\left(\hat w_j \geq n_{\max} \cdot \pi_j \right)\\
    = &~ \bar u_k(\hat w_{G_k}) + {r_k} \cdot \log\left(1+\epsilon_i\right) - {r_k} \cdot \epsilon_i\\
    < &~ \bar u_k(\hat w_{G_k}).
\end{align*}
This contradicts with the fact that $\wt w_{G_k}$ is the maximizer of $\bar u_k$. 

As a result,  $\frac{(\sum_{i \notin G_k}\pi_i^\ast)\wt w_i}{\sum_{i \notin G_k}\wt w_i} \geq \pi_i^\ast $ holds for any $i \in [n]$. Examining this inequality for all $i \in [n]$ yields that $\wt w = \lambda \cdot \pi^\ast$ for some $\lambda$ which must be greater than or equal to $n_{\max}$. 

\paragraph{Step 3.} We show that $\lambda = n_{\max}$.

From the definition of $n_{\max}$, for any $\lambda > n_{\max}$ there exists $k \in [K]$ such that 
\begin{align*}
    \left(u^{(k)}\circ \mec_{\eff}^{(k)}\right)(\lambda \cdot \pi^\ast) <&~ v_\ast^{(k)}\\
    \leq &~ \left(u^{(k)}\circ \mec_{\eff}^{(k)}\right)(n_{\max} \cdot \pi^\ast)\\
    \leq &~ \left(u^{(k)}\circ \mec_{\eff}^{(k)}\right)\left((n_{\max} \cdot \pi^\ast_{G_k}, \lambda \cdot \pi^\ast_{G_k^c})\right).
\end{align*}
This means that $\lambda \cdot \pi^\ast_{G_k}$ is not a NE, yielding a contradiction. It follows that $\left(n_{\max} \cdot \pi^\ast_{G_k}\right)_{k \in [K]}$ is the unique pure NE.

\end{proof}
% \fi

\clearpage
\section{Supporting Lemmas}

\begin{lemma}\label{lem:log_det_gradient}
We abbreviate $A^{(k)}$ and $A^{(k)}(A^{(k)})^\top$ as $A$ and $P$ respectively. Suppose $\M(w) \succ 0$. For any $l \in G_k$, we have
\begin{align*}%\label{eq:u_gradient}
    \frac{\partial - \log\det \left(A^\top \M(w)^{-1} A\right)}{\partial w_l} = \left\langle x_lx_l^\top , \M(w)^{-1} \right\rangle > 0.
\end{align*}
Furthermore, fixing $w_{G_k^c}$, the function $f(w) = -\log\det \left(A^\top \M(w)^{-1} A\right)$ is concave in $w_{G_k}$.
\end{lemma}

\begin{proof}
Indeed, for $l \in G_k$
\begin{align*}
    &~ \frac{\partial - \log\det \left(A^\top \M(w)^{-1} A\right)}{\partial w_l}\\
    = &~ - \left\langle \left(A^\top \M(w)^{-1} A\right)^{-1}, \frac{\partial A^\top \M(w)^{-1} A}{\partial w_l} \right\rangle\\
    = &~ - \left\langle \left(A^\top \M(w)^{-1} A\right)^{-1}, A^\top\frac{\partial \M(w)^{-1} }{\partial w_l} A \right\rangle\\
    = &~ - \left\langle \left(A^\top  \M(w)^{-1} A\right)^{-1}, A^\top \left(\sum_{i,j \in [d]} - \M(w)^{-1} e_ie_j^\top \M(w)^{-1} ( x_lx_l^\top)_{i,j} \right)A \right\rangle\\
    = &~ - \left\langle \left(A^\top  \M(w)^{-1} A\right)^{-1}, A^\top \left(- \M(w)^{-1} x_lx_l^\top \M(w)^{-1} \right) A \right\rangle\\
    = &~ \left\langle \left(A^\top  \M(w)^{-1} A\right)^{-1}, A^\top  \M(w)^{-1} A A^\top x_lx_l^\top \M(w)^{-1} A \right\rangle\\
    = &~ \left\langle I, A^\top x_lx_l^\top \M(w)^{-1} A \right\rangle\\
    = &~ \left\langle  x_lx_l^\top A A^\top, \M(w)^{-1} \right\rangle\\
    = &~ \left\langle x_lx_l^\top , \M(w)^{-1} \right\rangle,
\end{align*}
where the first step uses the fact that $\frac{\partial }{\partial Y_{i,j}} \log \det Y = [Y^{-1}]_{j,i}$; the third step uses the fact that $\frac{\partial }{\partial Y_{i,j}} Y^{-1} = -Y^{-1} e_i e_j^\top Y^{-1}$; the fourth step comes from 
\begin{align*}
    \sum_{i,j \in [d]} e_ie_j^\top( x_lx_l^\top)_{i,j} = \sum_{i,j \in [d]} e_ie_i^\top ( x_lx_l^\top) e_j e_j^\top = x_lx_l^\top;
\end{align*}
the fifth and final step use the fact that $A A^\top x_lx_l^\top = x_lx_l^\top $ since $AA^\top = P$ is the projection matrix on $\{x_i\}_{i \in G_k}$.

This establishes the first statement. 
To show concavity, notice that
\begin{align*}
    \left[-\mathrm{Hess}_f(w)\right]_{i,j} = &~ -\frac{\partial \left\langle x_ix_i^\top , \M(w)^{-1} \right\rangle}{\partial w_j}\\
    = &~ \left\langle x_ix_i^\top , \M(w)^{-1} x_jx_j^\top \M(w)^{-1} \right\rangle\\
    = &~ \left(x_i^\top \M(w)^{-1} x_j\right)^2, ~\forall i,j \in G_k.
\end{align*}
Therefore $-\mathrm{Hess}_f(w)$ is Hadamard product of the positive semi-definite matrix $\left(x_i^\top \M(w)^{-1} x_j\right)_{i,j \in G_k}$ and itself. 
It follows from Schur product theorem that the negative Hessian matrix $-\mathrm{Hess}_f$ is symmetric positive semidefinite everywhere in the domain, which establishes the concavity.
\end{proof}

\begin{lemma}\label{lem:incentive_compatibility}
For any $w$ such that $\sum_{i=1}^n w_i x_ix_i^\top$ is non-singular, 
\begin{align*}
    \log \det \left((A^{(k)})^\top \left(\sum_{i \in G_k} w_i x_ix_i^\top \right) A^{(k)} \right)
    \leq &~  \log \det \left( (A^{(k)})^\top\left( \sum_{i=1}^n w_i x_ix_i^\top  \right)^{-1}A^{(k)}\right)^{-1}.
\end{align*}\
Further, if equality holds, then for any $w_{G_k}'$ we have
\begin{align*}
    \left(A^{(k)})^\top (\sum_{i \in G_k} w_i' x_ix_i^\top ) A^{(k)} \right)^{-1}
     = (A^{(k)})^\top\left( \sum_{i \in G_k} w_i' x_ix_i^\top + \sum_{i\notin G_k} w_i x_ix_i^\top  \right)^{-1}A^{(k)}.
\end{align*}
\end{lemma}

\begin{proof}
Fix $k \in [K]$. We use shorthand notation $M_0 = \sum_{i \in G_k} w_i x_ix_i^\top, M = \sum_{i=1}^n w_i x_ix_i^\top, A = A^{(k)}$. Let $A = U \Lambda V^\top$ denote the Singular Value Decomposition (SVD) decomposition of $A$, where $U \in \R^{d \times d}$ and $V \in \R^{r_k \times r_k}$ are real orthogonal matrices. Then $M_0$ can be written as $M_0 = U \begin{pmatrix}
    D &\mathbf{0}_{r_k \times (d-r_k)}\\\mathbf{0}_{(d-r_k) \times r_k}&\mathbf{0}_{(d-r_k) \times (d-r_k)}
\end{pmatrix} U^\top$ where $D \in \mathbb{S}^{r_k}_+$. We write $U(M - M_0)U^\top = \begin{pmatrix}
    A_{11} & A_{12}\\A_{21} & A_{22}
\end{pmatrix}$ where $A_{11} \in \R^{r_k \times r_k}$. 

Notice that
\begin{align*}
    A^\top M_0 A = &~ V \Lambda^\top U^\top U \begin{pmatrix}     D &\mathbf{0}_{r_k \times (d-r_k)}\\\mathbf{0}_{(d-r_k) \times r_k}&\mathbf{0}_{(d-r_k) \times (d-r_k)} \end{pmatrix} U^\top U \Lambda V^\top\\
    = &~ V \begin{pmatrix}\mathbf{I}_{r_k \times r_k}& \mathbf{0}_{r_k \times (d - r_k)}\end{pmatrix}\begin{pmatrix}     D &\mathbf{0}_{r_k \times (d-r_k)}\\\mathbf{0}_{(d-r_k) \times r_k}&\mathbf{0}_{(d-r_k) \times (d-r_k)} \end{pmatrix}\begin{pmatrix}\mathbf{I}_{r_k \times r_k}\\ \mathbf{0}_{r_k \times (d - r_k)}\end{pmatrix} V^\top\\
    = &~ V D V^\top.
\end{align*}

We assert that $A^\top M^{-1} A \preceq \left(A^\top M_0 A\right)^{-1}$. Indeed, we have
\begin{align*}
    &~ \left(A^\top M_0 A\right)^{1/2}A^\top M^{-1} A\left(A^\top M_0 A\right)^{1/2}\\
    = &~ V D^{1/2} V^\top V \Lambda^\top U^\top U \begin{pmatrix}
    D + A_{11} & A_{12}\\A_{21} & A_{22}
\end{pmatrix}^{-1} U^\top U \Lambda V^\top V D^{1/2} V^\top\\
= &~ V D^{1/2}\begin{pmatrix}\mathbf{I}_{r_k \times r_k}& \mathbf{0}_{r_k \times (d - r_k)}\end{pmatrix} \begin{pmatrix}
    (D + A_{11} -A_{12}A_{22}^{-1}A_{21})^{-1} & \star \\ \star &\star 
\end{pmatrix} \begin{pmatrix}\mathbf{I}_{r_k \times r_k}\\ \mathbf{0}_{r_k \times (d - r_k)}\end{pmatrix} D^{1/2} V^\top\\
= &~ V D^{1/2}\begin{pmatrix}
    D + A_{11} -A_{12}A_{22}^{-1}A_{21}
\end{pmatrix}^{-1} D^{1/2} V^\top\\
\leq &~ V V^\top\\
= &~ \mathbf{I}_{r_k \times r_k}
\end{align*}
where the inequality is due to $A_{11} -A_{12}A_{22}^{-1}A_{21} \succeq 0$ since this is the Schur complement of $M - M_0 \succeq 0$. 
As a result, $A^\top M^{-1} A \preceq \left(A^\top M_0 A\right)^{-1}$. 

Applying the monotonicity of $\log \det(\cdot)$, we have
\begin{align*}
    \log \det \left( (A^{(k)})^\top\left( \sum_{i=1}^n w_i x_ix_i^\top  \right)^{-1}A^{(k)}\right) \leq \log \det \left((A^{(k)})^\top \left(\sum_{i \in G_k} w_i x_ix_i^\top \right) A^{(k)} \right)^{-1}.
\end{align*}
This establishes the inequality. 
Further, if equality holds then $A_{11} -A_{12}A_{22}^{-1}A_{21} = 0$. As a result,
\begin{align*}
    &~ A^\top M^{-1} A\\
    = &~ V \Lambda^\top U^\top U \begin{pmatrix}
    D + A_{11} & A_{12}\\A_{21} & A_{22}
\end{pmatrix}^{-1} U^\top U \Lambda V^\top \\
= &~ V \begin{pmatrix}\mathbf{I}_{r_k \times r_k}& \mathbf{0}_{r_k \times (d - r_k)}\end{pmatrix} \begin{pmatrix}
    (D + A_{11} -A_{12}A_{22}^{-1}A_{21})^{-1} & \star \\ \star &\star 
\end{pmatrix} \begin{pmatrix}\mathbf{I}_{r_k \times r_k}\\ \mathbf{0}_{r_k \times (d - r_k)}\end{pmatrix}  V^\top\\
= &~ V D^{-1} V^\top\\
= &~ \left(A^\top M_0 A\right)^{-1}.
\end{align*}

Notice that the validity above argument does not depend on $w_{G_k}$ and in fact holds for any $w_{G_k}'$.  
This completes the proof.

\end{proof}

\begin{claim}[Kakutani's fixed point theorem]\label{cla:fixed_point_thm}
Consider a set-valued function $F:\mathcal{C} \rightarrow 2^{\mathcal{C}}$ over convex and compact set $\mathcal{C}$ such that (i) $F$ has closed graph; (ii) $F(x)$ is non-empty and convex for any $x \in \mathcal{C}$, then there exists a fixed point $x$ such that $x \in F(x)$.
\end{claim}

\begin{claim}[Monotonicity of determinant]\label{cla:monotonicity_determinant}
Suppose $A$ and $B$ are two symmetric matrices such that $A \succeq B \succ 0$, then $\det A \geq \det B$. 
\end{claim}

\begin{claim}[Concavity of log-determinant function]\label{cla:concavity_log_determinant}
Suppose $A$ and $B$ are two symmetric positive semidefinite matrices such that $A \succeq B \succ 0$, then $\log \det (\lambda A + (1-\lambda)B) \geq \lambda \log \det B + (1-\lambda) \log \det A$ holds for any $\lambda \in (0,1)$.
\end{claim}

The following important result of Kiefer and Wolfowitz established the equivalence of the D-optimal design and G-optimal design.
\begin{theorem}[General equivalence theorem of G-optimal design \cite{kiefer1960equivalence}]\label{thm:equivalence_kiefer_wolfowitz}
Assume $\mathrm{span}(\X) = \R^d$. The followings are equivalent: 
\begin{itemize}
    \item $\pi^\ast = \underset{\pi \in \Delta(\X)}{\arg \max}~ \log \det \M(\pi)$;
    \item $\pi^\ast = \underset{\pi \in \Delta(\X)}{\arg \min}~\max_{x \in \X} x^\top \M(\pi)^{-1} x$;
    \item $\max_{x \in \X} x^\top \M(\pi)^{-1} x = d$.
\end{itemize}
\end{theorem}

\clearpage

\end{document}